\documentclass{article}
\usepackage{amsmath,amsthm,amsfonts,graphicx}
\usepackage{multirow}
\usepackage{multirow}
\def\smallddots{\mathinner{\raise7pt\hbox{.}\raise4pt\hbox{.}\raise1pt\hbox{.}}}
\def\smallsdots{\mathinner{\raise1pt\hbox{.}\raise4pt\hbox{.}\raise7pt\hbox{.}}}

\numberwithin{equation}{section}
\numberwithin{table}{section}
\newtheorem{theorem}{Theorem}[section]

\newtheorem{example}{Example}[section]
\newtheorem{definition}{Definition}[section]

\newtheorem{remark}{Remark}[section]

\usepackage{graphicx}
\graphicspath{%
    {converted_graphics/}%
    {/}%
}
\newcommand*{\ZZ}{\ensuremath{\mathbb{Z}}}
\usepackage{xspace}

\newcommand{\fgemm}{\texttt{fgemm}\xspace}
\newcommand{\pfgemm}{\texttt{pfgemm}\xspace}
\newcommand{\dgemm}{\texttt{dgemm}\xspace}
\newcommand{\sgemm}{\texttt{sgemm}\xspace}
\newcommand{\openblas}{\texttt{OpenBLAS}\xspace}
\newcommand{\plasmaquark}{\texttt{Plasma-Quark}\xspace}
\newcommand{\linbox}{\texttt{LinBox}\xspace}
\newcommand{\fflas}{\texttt{FFLAS-FFpack}\xspace}
\newcommand{\openmp}{\texttt{OpenMP}\xspace}

\newcommand{\strassenwinograd}{WRB-MM\xspace}
\newcommand{\OpenDreamKit}{the \href{http://opendreamkit.org}{OpenDreamKit} \href{https://ec.europa.eu/programmes/horizon2020/}{Horizon 2020} \href{https://ec.europa.eu/programmes/horizon2020/en/h2020-section/european-research-infrastructures-including-e-infrastructures}{European Research Infrastructures} project (\#\href{http://cordis.europa.eu/project/rcn/198334_en.html}{676541})}
\usepackage{amsfonts,amsmath}
\newcommand{\F}{\ensuremath{{\mathbb F}}}
\newcommand{\BigO}[1]{\ensuremath{\mathcal{O}\left(#1\right)}}
\usepackage{color,svgcolor}
\usepackage[plainpages=true]{hyperref}
\makeatletter
\hypersetup{
pdftitle={Fast Matrix Multiplication and Symbolic Computation},
pdfauthor={J-G. Dumas, V.Y. Pan},
breaklinks=true, 
colorlinks=true,
 linkcolor=darkred,
 citecolor=blue,
 urlcolor=darkgreen,
naturalnames=true}
\makeatother
\begin{document}

\centerline{{\Large \bf Fast Matrix Multiplication}}
\centerline{{\Large \bf and Symbolic Computation}}  


\medskip

\medskip

\centerline{Jean-Guillaume Dumas$^{[a]}$ and Victor Y. Pan$^{[b]}$}

\medskip
\medskip

\centerline{$^{[a]}$~Laboratoire Jean Kuntzmann}
\centerline{Applied Mathematics and Computer Science}
\centerline{Universit\'e Grenoble Alpes}
\centerline{700 avenue centrale, IMAG - CS 40700}
\centerline{38058 Grenoble cedex 9, FRANCE}
\centerline{\href{mailto:Jean-Guillaume.Dumas@imag.fr}{Jean-Guillaume.Dumas@imag.fr}}
\centerline{\href{http://ljk.imag.fr/membres/Jean-Guillaume.Dumas/}{ljk.imag.fr/membres/Jean-Guillaume.Dumas}\footnote{Partly
    supported by \OpenDreamKit.}}

\medskip
\medskip

\centerline{$^{[b]}$~Department of Mathematics and Computer Science}

\centerline{Lehman College of the City University of New York}

\centerline{Bronx, NY 10468 USA}

\centerline{and}
\centerline{Ph.D. Programs in Mathematics  and Computer Science}
\centerline{The Graduate Center of the City University of New York}
\centerline{New York, NY 10036 USA}
\centerline{\href{mailto:victor.pan@lehman.cuny.edu}{victor.pan@lehman.cuny.edu}}
\centerline{\href{http://comet.lehman.cuny.edu/vpan}{comet.lehman.cuny.edu/vpan}\footnote{Supported
    by NSF Grants CCF--1116736 and CCF--1563942 and PSC CUNY Award
    68862--00 46.}}

\begin{abstract}
The complexity of matrix multiplication (hereafter MM) 
 has been intensively  studied since 1969,
when 
 Strassen surprisingly decreased the 
exponent 3 in the cubic cost of the  straightforward classical MM 
 to $\log_2(7)\approx 2.8074$.
Applications to some  fundamental problems of Linear Algebra and Computer Science
have been immediately recognized, but  the
researchers in Computer Algebra keep discovering more and more
applications 
even today, with
 no sign of slowdown.
We survey the unfinished history of decreasing the exponent
towards its information lower bound 2,
recall some important techniques discovered in this process
and linked to other fields of computing,
reveal sample surprising applications
to fast computation of the inner products
of two vectors
and summation of integers,
 and
discuss the {\em curse of recursion}, which separates the progress in fast MM into 
its most acclaimed and purely theoretical part and into  
 valuable acceleration of MM of feasible  sizes.
Then, in the second part of our paper, we cover fast MM
  in realistic symbolic computations and
discuss 
applications
and implementation of fast
\emph{exact} matrix multiplication. 
We first review how most 
of exact linear algebra can be reduced to matrix
multiplication over small finite fields. Then we highlight the differences in
the design of approximate and exact implementations of fast MM, taking into account
nowadays processor and memory hierarchies. In the concluding section we comment
on current perspectives of the study  of fast MM.
\end{abstract}

\paragraph{\bf 2000 Math. Subject Classification:}
68Q25, 65F05, 15A06,
15A69, 01A60, 15-03

\paragraph{\bf Key Words:}
Matrix multiplication (MM),
Computations, 
Complexity,
Linear algebra,
Tensor decomposition,
Multilinear algebra,
Inner product,
Summation,
Binary Segmentation,
Exact linear algebra,
Small finite fields,
Processors, 
Memory hierarchy,
Impacts of fast MM,
Numerical implementation,
Perspectives

\section{Introduction}\label{sintr}

\subsection{Our subjects}\label{ssbj}

 Matrix multiplication 
(hereafter we keep using the acronym {\em MM}) is fundamentally important 
for symbolic and numerical computations in linear algebra and for
the theory of computing.    
Efficient performance of MM
depends on various factors, 
particularly on vectorization, data locality, and arithmetic cost
(cf.\ \cite[Chapter 1]{GL13}).

In the first part of the paper (Sections \ref{sst69}--\ref{sapplcs}) we review  the work on the decrease of 
the arithmetic cost, including purely theoretical study of MM of immense sizes
(so far this part of the study has been most acclaimed and most generously supported!),  but we focus on
feasible MM.

In our longest Section \ref{sec:exactcomp} 
we discuss 
realistic acceleration of 
symbolic MM, taking into account
nowadays processor and memory hierarchies.

In our concluding Section \ref{sprimp} we comment
on 
current perspectives of the study  of fast MM.

\subsection{History of fast MM and its impacts}\label{shstr}

The cubic arithmetic time $2n^3-n^2$ of 
the straightforward algorithm for $MM(n)$, that is,
for $n\times n$ MM,
was commonly believed to be optimal until 1969, when Strassen's algorithm
 of \cite{S69}
performed $MM(n)$ in $O(n^{\omega})$ time for $\omega=\log_2(7)\approx 2.8074$.
This implied the  exponent   
 $\log_2(7)$ also for numerous  venerated computational problems
in 
Computer Science, Linear Algebra,  and Computer  Algebra
such as Boolean MM, 
parsing
context-free grammars, computing paths and distances in graphs,
 the solution of a nonsingular 
linear system of  equations, 
 computation of the inverse and the determinant
of a matrix, and
its various factorizations 
 (see more in Section~\ref{sapplcs}).
The worldwide interest to MM has
immediately exploded,\footnote{For the scientific world the news came as a
miracle from the blue. Most of the readers were particularly impressed by the power of 
the divide and conquer method (not novel in 1969) rather than by Strassen's ingenious
algorithm for $2\times 2$ MM, and many scientists, although not experts like Strassen,
ignored or overlooked a minor but meaningful earlier acceleration of the  straightforward
MM that saved about 50\% of its scalar multiplications (see Example~\ref{exw0} in Section~\ref{sst69}).}
and it was widely expected that  
new efficient algorithms would soon perform MM and solve 
the related computational problems in nearly
quadratic time. Even the exponent 2.8074, however, defied
the attacks of literally all experts around the globe for almost a decade, until 
1978, when the algorithm of \cite{P78} 
broke Strassen's record, improving the algorithm of \cite{S69}
already at the level of feasible MM.
 
 The mainstream research
responded to that breakthrough
by directing all effort to the decrease of
 the exponent of MM of
 unrestricted sizes 
and very soon succeeded 
in dramatic acceleration of infeasible MM of 
 astronomical sizes. 
New surprising resources have been found, 
  sophisticated 
techniques have been developed, and
by 1987 the exponent of infeasible MM was decreased  
 below 2.38 (see \cite{CW90}),
although as of December 2016 it still has not been decreased below 2.37, that is,  
 in the last  3 decades the progress was nominal
(see \cite{LG14} for the current  record exponent). 
Moreover, the study of infeasible MM has never
made impact on practice of MM or any
other realistic computations 
(cf.  \cite{ASU13}, \cite{AFLG14}, \cite{BCCG16}, and our concluding section).

For $n$  restricted to be ``moderate",
 say, less than 1,000,000,  the current record  
 is 2.7734,
achieved with the algorithm  of \cite{P82} and
 unbeaten since 1982.
 All algorithms supporting smaller exponents 
suffer from the {\em curse of recursion} (cf. \cite{P14a}):
they beat the classical straightforward MM algorithm only 
after performing a large and typically immense number 
of recursive steps, with the input size 
growing exponentially in the number of such steps:
 the straightforward 
algorithm supersedes them until the input size
by far exceeds realistic level, typically by many orders of magnitude.

\subsection{Focus of our presentation}\label{sfcs}

In the context of this development,
 we refocus our presentation compared to  the decades-old survey
 \cite{P84}. In Section~\ref{sfrth}  we still pay 
tribute to the  lasting interest to the exponent
 of infeasible MM, but we do not cover 
various amazing
sophisticated techniques proposed 
{\em exclusively for the
 acceleration of MM of immense sizes}, which
dominated the review of \cite{P84}.
Instead we
 cover in some detail
the 
techniques
that
are efficient already for MM of moderate sizes and 
 have  impacts
on realistic computations beyond MM. 
We feel that 
reduction of various  
 computational problems  to MM
is interesting on its own right 
and because of potential benefits of wider application of
fast or even straightforward algorithms for feasible MM.
Lately the study of these links 
was particularly intensive in
the field of symbolic computation (see, e.g., Proceedings of ISSAC 2015 and ISSAC 2016).

We recall that historically no adequate  comprehensive  review 
of the MM subject has appeared
for  decades,  not to the benefit of the field.
As we already explained,   
after the
breakthrough of 1978, public 
interest  to fast  feasible MM was diverted 
by worldwide excitement about the exponent of infeasible MM, but
also by some other factors.
In particular the advanced techniques of  of \cite{P78} 
were much harder  for non-experts to grasp than the 
catchy  divide and conquer method, 
and public 
attention to the fast algorithm of \cite{P78} for feasible MM was also hurt
by  the folk ``theorem" about its alleged numerical instability.
This ``theorem" has somehow
spread fast throughout  the communities of numerical and symbolic linear algebra,
before the classical paper \cite{BL80} of 1980 proved that
the ``theorem" was false.\footnote{More precisely fast MM
 algorithms are slightly less
 stable numerically than the straightforward MM,
but this instability is  mild and rather little affects actual implementations of MM (see more details
in \cite{BL80}, \cite{H90}, \cite{BDHS12}, \cite{DDHK07}, and  \cite{BBDLS15}).} The results of  \cite{BL80} became widely known only
when the well-recognized article \cite{DDHK07} extended them  to all
recursive bilinear algorithms for MM, but
even in 2010 the Introduction of the important innovative paper \cite{Bodrato:2010:SMM}
still referred to this ``theorem"%
and in 2016, the paper \cite{HSHG16} still talks about ``numerical stability
issues with many levels of recursions''\footnote{The paper
\cite{HSHG16} is not really about fast MM since it unrolls only one
or two levels of recursion of Strassen's algorithm and then
recomputes several times the submatrix additions to avoid using
temporary buffers.}.   

Moreover
the paper \cite{DDHK07} and its successor \cite{BDHS12}
 attack \cite{P78} as well as all the work on fast feasible MM 
from another side. Trying to simplify the discussion
or perhaps to divert public attention from the advanced work on fast feasible MM 
to the domain of their own study and progress, the authors of these papers 
 call ``Strassen-like" 
 all known fast algorithms for MM.
This ``innovation" was based on ignorance:  ``Strassen-like" algorithms, as 
 \cite{DDHK07} and  \cite{BDHS12} formally define them,
have been long and well known under the
much more informative name of {\em noncommutative bilinear algorithms}
(see, e.g., \cite{BM75}).\footnote{The book \cite{P84b} and MM survey articles  \cite{P84} and  \cite{P84a} 
 pay high respect to Strassen's fundamental contributions to fast MM (and similarly did Volker Strassen 
to the contribution of \cite{P66} to algebraic computations in 
sections ``Pan's method" of \cite{S72} and \cite{S74}), 
but we feel that calling all the advanced work on fast feasible MM  Strassen-like,
is not more fair or informative than, say, labeling  Democritus-like
the Faraday's constant, Mendeleev's periodic table, and
 the Heisenberg's principle of
 uncertainty.}
 Our personal communication in 2015 with the authors of \cite{DDHK07}  seems to help:
the label ``Strassen-like" is not used, e.g., in \cite{BBDLS15},
but unfortunately their original widely publicized  contempt to the advanced results on fast feasible MM 
(including those completely distinct from
 Strassen's old algorithm of 1969 and in various important respects superseding it)
has been deeply implanted into scientific community.

In the first part of 
our paper (Sections \ref{sst69}--\ref{sapplcs})
we review the previous study of fast MM 
with the focus 
on fast  feasible MM and
  its impacts and applications 
to realistic computations beyond MM, and for example we
included our novel extension of an old MM technique
  to the computation of the inner product of 2 vectors
and the summation of integers 
(see our Examples~\ref{exinner} and~\ref{exsum}). 

In the second part of the paper (Section \ref{sec:exactcomp}) we 
discuss in some detail symbolic  implementation 
of fast MM directed to minimizing communication cost
and improving parallel implementation.  
The basic algorithms for that study are mostly decades-old,
and complementing them with some more recent advanced 
algorithms for feasible MM   is one of the most natural 
directions to further progress in the field.

\subsection{Organization of our paper}\label{sorgn}

We organize our paper as follows. In Sections~\ref{sst69}
and~\ref{strex} we recall the 2 first 
accelerations of MM, 
in 1969 and 1978, respectively, and comment on 
their impacts beyond MM.  
In Sections
\ref{sblnr},~\ref{strld}, and~\ref{sapa} 
we cover the fundamental classes of 
 bilinear, trilinear and the so called APA algorithms, respectively,
and discuss the associated fundamental techniques of the 
algorithm design, their impact on the acceleration of MM
 and their links to other areas of computing
and algebra.
In Section~\ref{sbsg} we extend the APA  
technique to  
computing the inner product of 2 vectors  and summation. 
In Section~\ref{sfrth} we summarize the history of 
 fast MM after 1978. In Section~\ref{sapplcs} we further 
comment on applications of fast MM and the impact of its study  to other areas of
computing and Mathematics.  
In Section~\ref{sec:exactcomp} we discuss 
applications and implementations of {\em exact} fast MM.
In Section~\ref{sprimp} we
comment on numerical implementation of fast MM
and  perspectives of its study.

 In addition to the acronym $``MM"$ for ``matrix multiplication", 
hereafter we write $``MI"$  for ``nonsingular matrix inversion", 
 $MM(m,n,p)$  for $m\times n$ by $n\times p$ MM,
$MM(n)$  for $M(n,n,n)$, and
 $MI(n)$  for $n\times n$ $MI$.

 $W=(w_{i,j})_{i,j=1}^{m,n}$
 denotes an $m\times n$ matrix with the entries 
$w_{i,j}$, for $i=1,\dots, m$ and $j=1,\dots, n$.

\section{1969: from the Exponent 3 to 2.8074 by means of \texorpdfstring{$2{\times}2$}{2x2}-based recursive processes}\label{sst69}

We begin with recalling 3 old accelerations of the  straightforward MM algorithm.

\begin{example}\label{exw0}  {\em From faster inner product to faster MM.} [See \cite{W68}
and notice technical similarity to the algorithms for polynomial evaluation with 
preprocessing in \cite{P66} and \cite{K97}.]
Observe that
\begin{equation}\label {eqcmmm}
{\bf u}^T{\bf v}=\sum_{i=1}^{n/2}(u_{2i-1}+v_{2i})(v_{2i-1}+u_{2i})-
\sum_{i=1}^{n/2}u_{2i-1}u_{2i}-\sum_{i=1}^{n/2}v_{2i-1}v_{2i},
\end{equation}
 for any even $n$. 
 Apply this 
 identity to all $n^2$ inner products 
defining $n\times n$ MM 
and compute it by using  $0.5  n^3 +n^2$ 
scalar multiplications
and $1.5n^3+2n^2-2n$  additions and subtractions.
\end{example}

\begin{example}\label{ex1} {\em Winograd's $2\times 2$ MM}
(cf.\  \cite{F74}, \cite[pages 45--46]{BM75},
 \cite[Exercise 6.5]{AHU74}, or \cite{DHSS94}).

Compute the product $X=UV$ of a pair of $2\times 2$ matrices, 
$$U=\begin{pmatrix} 
u_{11} & u_{12} \\ 
u_{21} & u_{22}
 \end{pmatrix},~V=\begin{pmatrix} 
v_{11} & v_{12} \\ 
v_{21} & v_{22}
 \end{pmatrix},~X=UV=\begin{pmatrix} 
x_{11} & x_{12} \\ 
x_{21} & x_{22}
 \end{pmatrix}$$ by using the following expressions, 

$s_1=u_{21}+u_{22}$, 
$s_2=s_1-u_{11}$,
$s_3=u_{11}-u_{21}$,
$s_4=u_{12}-s_2$,

$s_5=v_{12}-v_{11}$,
$s_6=v_{22}-s_5$, 
$s_7=v_{22}-v_{12}$,
$s_8=s_6-v_{21}$, 

$p_1=s_2s_6,~p_2=u_{11}v_{11},~p_3=u_{12}v_{21}$,  \\
$p_4=s_3s_7,~p_5=s_1s_5,~p_6=s_4v_{22}$, 
$p_7=u_{22}s_8$,  

$t_1=p_1+p_2,~t_{2}=t_1+p_4$, $t_{3}=t_1+p_5$,  \\
$x_{11}=p_2+p_3,~x_{12}=t_3+p_6,~x_{21}=t_2-p_7,~x_{22}=t_2+p_5$.
\end{example}

This algorithm performs $2\times 2$ MM by
using 7 scalar multiplications and 15 scalar additions and subtractions
instead of 8 and 4, respectively,  that is, 22 versus 12.
Fix, however, a sufficiently large integer $h$,
then recursively $h-1$ times 
 substitute $2\times 2$
matrices for $u_{ij}$, $v_{jk}$, and $c_{ik}$,
for all subscripts $i,j,k\in\{1,2\}$, and reuse the above algorithm
for every  $2\times 2$ MM.
This computation only involves $7^h$ scalar multiplications 
and $7^{h-1}+15(4^h)$ additions and subtractions,
which is readily extended to performing $n\times n$ MM for any $n$
by using $c~ n^{\log_2(7)}$ scalar arithmetic operations 
 overall where $\log_2(7)\approx 2.8074$
and  careful refinement and analysis yield $c<3.92$
 \cite{F74}.

\begin{example}\label{ex0} 
{\em Strassen's $2\times 2$ MM}  \cite{S69}.

$p_1=(u_{11}+u_{22})(v_{11}+v_{22}),~p_2=(u_{21}+u_{22})v_{11},~p_3=u_{11}(v_{12}-v_{22})$,  \\
$p_4=(u_{21}-u_{11})(v_{11}+v_{12}),~p_5=(u_{11}+u_{12})v_{22},~p_6=u_{22}(v_{21}-v_{11})$,  
$p_7=(u_{12}-u_{22})(v_{21}+v_{22})$,  \\
$x_{11}=p_1+p_6+p_7-p_5,~x_{12}=p_3+p_5,~x_{21}=p_2+p_6,~x_{22}=p_1+p_3+p_4-p_2$.
\end{example}

This algorithm performs $2\times 2$ MM by
using 7 scalar multiplications and 18 scalar additions and subtractions,
which is a little inferior to Example~\ref{ex1},
but also initiates a similar recursive process that
performs  $n\times n$ MM for any $n$ 
by using $c'\cdot n^{\log_2(7)}$ scalar arithmetic operations
for $c'<4.54$. 
Being a little slower, this algorithm  is used much less than Winograd's,
but is much
more celebrated  because it appeared earlier.

In MM literature 
Winograd's algorithm is rarely called Winograd's, but usually ``Strassen--Winograd's algorithm"
or ``Winograd's variant of Strassen's algorithm",
and sometimes less competently {\em Strassen's}, {\em Strassen-based},
 or {\em Strassen-like} (cf. \cite{DDHK07}, \cite{BDHS12}).
We can see nothing in Example \ref{ex1} borrowed from Example \ref{ex0}, 
but the cited impact of \cite{DDHK07} seems to be so deeply rooted in the Community of Computer Algebra,
that even the authors of
 the advanced works \cite{Bodrato:2010:SMM} and \cite{CH16}
did not dare to call 
Example \ref{ex1} ``Winograd's  algorithm",
although their own innovative algorithms extended precisely this example 
and not  Example \ref{ex0}.

 The first line of Table \ref{tab2cmpl} displays the estimated 
arithmetic cost of the recursive bilinear algorithms for $MM(n)$, $n=2^k$,
that begin with $2\times 2$ MM by Strassen (of Example \ref{ex0}), Winograd
(of Example \ref{ex1}), and Cenk and Hasan 
(of \cite{CH16}).
The second line shows the decreased cost bounds where recursive process begins with 
$k\times k$ MM performed with straightforward algorithm, and then one of
the 3 above algorithms is applied recursively.
The improvement from \cite{CH16} is technically interesting, but
its arithmetic cost bounds are still significantly inferior 
to those of the algorithms  of \cite{P82} as well as ones of \cite{LPS92}, whose
 implementation  in \cite{K04} is numerically stable and is highly 
efficient in using memory space.

\begin{table}[ht] 
\caption{Arithmetic Complexity of Some $2\times 2$-based Recursive Bilinear Algorithms.}
\label{tab2cmpl}
  \begin{center}
\begin{tabular}{| c | c |c |}
      \hline
 Strassen's (k=10) & Winograd's (k=8)& Cenk and Hasan's (k=8) \\ \hline
 $7n^{2.81} -6n^2$ &  $6n^{2.81} -5n^2$   & $5n^{2.81} + 0.5n^{2.59} + 2n^{2.32} -6.5n^2$ \\ \hline
 $3.89n^{2.81} -6n^2$   & $3.73n^{2.81} -5n^2$ & $3.55n^{2.81} + 0.148n^{2.59} + 1.02n^{2.32} - 6.5n^2$ \\ \hline
 \end{tabular}
\end{center}
\end{table}

\section{Bilinear Algorithms}\label{sblnr}

\subsection{The class of bilinear algorithms}\label{sblnrdef}

The  algorithms of Examples~\ref{ex1} and~\ref{ex0} belong to the important 
class of noncommutative bilinear algorithms, to which we refer just
as {\em bilinear}. 
Such an algorithm for $MM(m,n,p)$
first computes some linear forms $l_q(U)$ and $l'_q(V)$ in the entries of 
the input matrices $U=(u_{ij})_{i,j=1}^{m,n}$ and $V=(v_{jk})_{j,k=1}^{n,p}$
and then  
the entries $x_{ik}=\sum_j u_{ij} v_{jk}$ of the product $X=UV$ as 
the $mp$ bilinear forms,
$$l_q(U)=\sum_{i,j=1}^{m,n}\alpha_{ij}^{(q)}u_{ij},~l'_q(V)=\sum_{j,k=1}^{n,p}\beta_{jk}^{(q)}v_{jk},~q=1,\dots,r,$$
$$x_{ik}=\sum_{q=1}^r\gamma_{ik}^{(q)}l_q(U)l'_q(V),~~i=1,\dots,m;~k=1,\dots,p.$$
Here $r$ is said to be the {\em rank of the algorithm}, 
$u_{ij}$, $v_{jk}$, and $x_{ik}$ are  variables
or  block matrices,  and
$\alpha_{ij}^{(q)}$, $\beta_{jk}^{(q)}$, and $\gamma_{ik}^{(q)}$  are constant
coefficients for all $i,j,k$, and $q$. They define coefficient tensors 
$(\alpha_{ij}^{(q)})_{i,j,q}$, $(\beta_{jk}^{(q)})_{j,k,q}$, and $(\gamma_{ik}^{(q)})_{i,k,q}$,
which can be also viewed as coefficient matrices 
\begin{equation}\label{eqcoeff}
A=(\alpha_{ij}^{(q)})_{(i,j),q},~B=(\beta_{jk}^{(q)})_{(j,k),q},~{\rm and}~C=(\gamma_{ik}^{(q)})_{(i,k),q}.
\end{equation}

\subsection{The ranks of bilinear algorithms and the \texorpdfstring{$MM$}{MM} exponents}\label{sblnrrnk}

Suppose $m=n=p$, assume that 
$u_{ij}$, $v_{jk}$, and $x_{ik}$ are   block matrices,  
and then again
recursively apply the same bilinear algorithm of rank $r$
to block matrices. The resulting 
 bilinear  algorithms have ranks $r^{h}=c'n^{h\omega}$ 
for
$MM(n^h)$,
$h=2,3,\dots$,  
$\omega=\omega_{n,r}=\log_{n}(r)$ and a constant $c'$ is independent of $n$ and $h$.
One can readily
extend these observations to 
the following result.

\begin{theorem}\label{thexpnr}
Given a bilinear algorithm of rank $r$ for $MM(n)$
for a pair of positive integers $n$ and $r$,
one can perform 
$MM(K)$ by using $cK^{\omega}$ arithmetic operations for 
any $K$, $\omega=\omega_{n,r}=\log_{n}(r)$,
and
a constant $c$ independent of $K$.
\end{theorem}

Now we define 

(i) the {\em exponent of MM(n)},
\begin{equation}\label{eqexpnmmn}
\omega_n=\min_r\log_{n}(r)
\end{equation}
where an integer $n>1$ is  fixed
and the minimum is over the ranks $r$ of all 
bilinear algorithms for $MM(n)$
and 

(ii) the {\em exponent of MM},
\begin{equation}\label{eqexpnmm}
\omega=\min_n\omega_n
\end{equation}
where the minimum is over all integers $n>1$.\footnote{Here we consider MM over the fields
of real and complex numbers. The exponent $\omega$ (although not the overhead constant)   
stays invariant over   
the  fields having the same characteristic \cite[Theorem 2.8]{S81}.
Some of our results 
can be used in important applications of MM over semi-rings, but generally 
in that case distinct techniques are used
(cf.  \cite{AGM97}, \cite{DP09}, \cite{Y09}, \cite{LG12}).} 

(iii)  The latter concept
 is mathematically attractive, but the highly rewarded work for its upper estimates
has diverted public attention from feasible to infeasible MM.
 To straighten the unbalanced study 
of this subject, one should instead
estimate the {\em exponents of feasible MM},
\begin{equation}\label{eqexpnmm<}
\omega_{<N}=\min_{n<N}\omega_n,
\end{equation}
for realistic upper bounds $N$ on the  dimension of MM inputs.

\subsection{Bilinear versus quadratic
algorithms for bilinear problems}\label{sblnrfrth}

The straightforward  algorithm for $MM(m,n,p)$ is bilinear of rank $mnp$. 
The bilinear algorithms of Examples~\ref{ex1} and~\ref{ex0} for $MM(2)$
have rank 7, 
which    turned out to be 
optimal (see \cite{HK69}, \cite{HK71}, 
   \cite{BD78}).
 \cite[Theorem 3]{P72} as well as  \cite[Theorem 0.3]{P14} and \cite{dG78})
provide an explicit expression for all bilinear algorithms of rank 7 for $MM(2)$,
including the algorithms of Examples~\ref{ex1} and~\ref{ex0}
as special cases. 
Among them 15 scalar additions and subtractions of Example~\ref{ex1}
are optimal 
\cite{P76}, \cite{B95}. 

One can define bilinear algorithms for any {\em bilinear problem}, that is, for
the computation of any set of bilinear forms, 
e.g.,  the product of two complex numbers 
$(u_1+{\bf i} u_2)(v_1+{\bf i} v_2)=(u_1v_1-u_2v_2)+{\bf i}(u_1v_2+u_2v_1)$, 
${\bf i}=\sqrt{-1}$.
The straightforward bilinear algorithm has rank 4, and here is a rank-3 bilinear algorithm,
$l_1l_1'=u_1v_1$, $l_2l_2'=u_2v_2$, $l_3l'_3=(u_1+u_2)(v_1+v_2)$,   
$u_1v_1-u_2v_2=l_1l_1'-l_2l'_2$, $u_1v_2+u_2v_1=l_3l_3'-l_1l_1'-l_2l_2'$. 
See \cite{W70}, \cite{F72}, \cite{F72a},
\cite{P72}, \cite{BD73}, \cite{HM73}, \cite{S73}, \cite{P74}, \cite{BD76}, \cite{BD78}, 
on the early study of bilinear algorithms
and see a concise exposition in \cite{BM75}. The book
\cite{W80} covers various  efficient
 bilinear algorithms for multiplication of pairs of integers and polynomials
(the latter operation is also called the {\em convolution} of the coefficient  vectors),
with further applications to the design of FIR-filters. 

The minimal rank of all bilinear algorithms
for a fixed bilinear 
problem such as $MM(m,n,p)$
is called the {\em rank of the problem}.
It can be bounded in terms of the minimal arithmetic cost
of the solution of the problem
 and vice versa
\cite[Theorem 1]{P72}, \cite{BM75}.
  
The algorithm 
 of Example~\ref{exw0} of rank $r=r(n)=0.5n^3 +n^2$ for $MM(n)$, however, 
is not (noncommutative) bilinear; such algorithms are
  called {\em quadratic}
or commutative bilinear. 
(See \cite{W71} on some lower bounds on
the rank of such algorithms for MM.)
We cannot extend to them  
recursive processes that bound the MM exponents by $\log_n(r)$,\footnote{Hereafter
we refer to  decreasing upper bounds 
on the exponent of MM as {\em decreasing the exponent}, for short.} 
because MM is not commutative, e.g.,
 the equation $u_{2i-1}u_{2i}=u_{2i}u_{2i-1}$  
is invalid for  matrices
$u_{2i-1}$ and $u_{2i}$.

The algorithm, however, was of great importance in the history of 
fast MM: it was the first acceleration of the straightforward MM
(it saves about 50\%
multiplications), but most important,
it motivated  the effort for the design of 
bilinear (rather than quadratic) algorithms of rank less than $n^3$ for $MM(n)$.
It ``remained" to devise such an algorithm at least for $n=2$,
and Strassen received ample recognition for his brilliant design 
that accomplished exactly this.

\section{1978: from 2.81 to 2.78 by means of trilinear aggregation}\label{strex}

In 1969 the exponents below Strassen's 2.8074 
became the target 
of literally all the leading researchers in the field, worldwide,
 but remained a dream
for almost a decade. 
This dream would have come true
based on bilinear algorithms of rank 6 for $MM(2)$ 
 or rank 21 for
$MM(3)$,
but it was proved that the rank of $MM(2)$ exceeds 6, 
and it is still unknown whether $MM(3)>22$. 

We refer the reader to
the paper 
 \cite{MR14} for
the current record lower bounds on the 
rank of $MM(n)$ for all $n$, to the papers
 \cite{HK69}, \cite{HK71},  \cite[Theorem 1]{P72}, and
\cite{BD73},   \cite{BD78},   
 \cite{B89},  \cite{B99},  \cite{B00},
 \cite{RS03}, \cite{S03}, and \cite{L14}
for some earlier work in  this direction, and to the papers
\cite{DIS11} and \cite{S13} for various lower and upper bounds
on the arithmetic complexity and the ranks of rectangular MM of smaller sizes.

The exponent was decreased only in 1978, after almost a decade of stalemate, 
 when
the paper \cite{P78} presented 
a bilinear algorithm of rank 143,640 for $MM(70)$. This breakthrough implied 
the exponent $\omega=\log_{70} (143,640)<2.7962$
for $MM(n)$, $MI(n)$, Boolean $MM(n)$, 
and a variety of other well-known computational problems.
The algorithm of \cite{P78} has
extended an algorithm of 
the paper \cite{P72} of 1972,  
published in Russian\footnote{Until 1976 the second author lived in the Soviet Union. 
From 1964 to 1976 he has been working
 in Economics in order
to make his living
and has written the papers \cite{P66} and \cite{P72} in his spare time.}
and translated into English only in 2014
in \cite{P14}. 

The progress was due to the novel combination of two techniques:
 {\em trilinear interpretation} of bilinear algorithms
and the {\em aggregation} method. By following \cite{P78}
we call this combination
{\em trilinear aggregation}.
By refining 
 this combination of 2 techniques
the paper
 \cite{P82} accelerated  MM of moderate sizes 
and
 yielded the exponent 2.7734. As we already mentioned, various algorithms 
combining trilinear aggregation
with other advanced techniques
  decreased the exponent 
 below this level
(and in 1986  even below  2.38), but 
 only when they 
were applied to MM of immense sizes because of the curse of recursion.

The technique of trilinear aggregation has 
been recognized for its impact
on the decreases of the MM exponent,
but the paper \cite{P72} was also a historical  
landmark in the study of multilinear and tensor decompositions.
Such decompositions  introduced by Hitchcock
in 1927  received little attention except for
a minor response in  1963--70 with half of a dozen
papers  in  
 the psychometrics literature. 
The paper \cite{P72} of 1972 
 provided the earliest known application of nontrivial 
multilinear and tensor decompositions to fundamental matrix computations,
now a popular flourishing area in
linear and multilinear algebra with a wide range of 
important applications to modern computing 
(see  \cite{T03}, \cite{KB09},  \cite{OT10}, 
\cite{GL13},
and the bibliography therein).
Nevertheless
the paper  \cite{P72} has rarely  been cited at all
and has never been cited in the papers on multilinear and tensor decompositions.

\section{Trilinear Decompositions and Duality}\label{strld}

Next we define trilinear representation of MM, first proposed and used in 
the paper \cite{P72}. 
We are also going to link it to some important
computations beyond MM.

Let $U=(u_{ij})_{i,j}$ and $V=(v_{jk})_{j,k}$ be a pair of  $m\times n$
and $n\times p$ matrices, respectively, and
let the equations
$\sum_j u_{ij} v_{jk}=\sum_{s=1}^rw_{ik}^{(s)}l_s(U)l'_s(B)$ for all $i,j$
represent 
 a bilinear algorithm  of rank $r$ 
for the matrix product $X=UV$.

Define  
a  {\em trilinear decomposition} of rank $r$
for $trace(UVW)=\sum_{i,j,k} u_{i,j} v_{jk}w_{ki}$
by multiplying  these equations
by  variables $d_{ki}$ and summing the products in $i$ and $k$.
Here $W=(w_{ki})_{k,i}^{n,m}$
is an auxiliary  $n\times m$ matrix and 
 trace$(M)$ denotes the trace of a matrix $M$.
(Equivalently we can decompose  
the tensor 
 of the trilinear form $trace(UVW)$ 
 into the sum of $r$ tensors of  rank 1.)

\begin{example}\label{extrmm2}
{\em A trilinear decomposition    
of rank 7 for $MM(2)$.}  \\
$~~~~~~~~~ \sum_{i,j,h=1}^2u_{ij}v_{jh}w_{hi}=\sum_{s=1}^7l_sl_s'l_s''$,~
$l_1l_1'l_1''=(u_{11}+u_{22})(v_{11}+v_{22})(w_{11}+w_{22})$, \\
$l_2l_2'l_2''=(u_{21}+u_{22})v_{11}(w_{21}-w_{22})$,
$l_3l_3'l_3''=u_{11}(v_{12}-v_{22})(w_{12}+w_{22})$,
$l_4l_4'l_4''=(u_{21}-u_{11})(v_{11}+v_{12})w_{22}$,\\
$l_5l_5'l_5''=(u_{11}+u_{12})v_{22}(w_{12}-w_{11})$,
$l_6l_6'l_6''=u_{22}(v_{21}-v_{11})(w_{11}+w_{21})$,
$l_7l_7'l_7''=(u_{12}-u_{22})(v_{21}+v_{22})w_{11}$.
\end{example}
Conversely, we can come back to 
the original bilinear algorithm
for  the matrix product $X=UV$ if we
interpret both sides of any decomposition 
of the  trilinear form $trace(UVW)$
as linear forms in the variables $w_{ih}$ and
equate the coefficients of these variables on both sides of the 
decomposition.
 
More generally,
 \cite[Theorem 2]{P72} 
states the equivalence of 
a bilinear algorithm of rank $r$ for $MM(m,n,p)$
 to a trilinear decomposition 
of rank $r$ for 
 the associated trilinear form and its tensor.

Instead of equating the variables $w_{ij}$
 on both sides 
of the trilinear decomposition, we can
 equate the coefficients of 
all variables $u_{ij}$ or all variables $v_{jh}$ 
and then arrive at 2 
other dual bilinear algorithms
of the same rank
for the problems $M(n,p,m)$ and $M(p,m,n)$. 

By interchanging the subscripts of the variables,
we arrive at the dual bilinear algorithms of the same rank  
for the problems $MM(m,p,n)$, $MM(n,m,p)$, and $MM(p,n,m)$ as well
(cf.\ \cite[part 5 of Theorem 1]{P72},  \cite{BD73},
\cite{HM73}, \cite{P74}). 
The latter extension from triples to 6-tuples 
is pertinent to MM,
because it uses the double subscripts 
for the variables,  
but the triples of bilinear algorithms
can be generated from their common 
trilinear representation 
for any bilinear computational problem, e.g., for
multiplication
of 2 complex numbers in the following example.

\begin{example}\label{extrcopr} 
{\em A trilinear
decomposition 
of rank 3 for 
 multiplication
of 2 complex numbers.} 
$$u_1v_1w_1-u_2v_2w_1+u_1v_2w_2+u_2v_1w_2=
u_1v_1(w_1-w_2)-u_2v_2(w_1+w_2)+(u_1+u_2)(v_1+v_2)w_2.$$
\end{example}
 
For a sample application of the duality technique, one can readily deduce 
the following result
 (part 1 of \cite[Theorem 1]{P72}).

\begin{theorem}\label{thdual} 
Given a bilinear or trilinear algorithm 
of rank $r$ for   $MM(m,n,p)$
and any 4-tuple of integers $r$, $m$, $n$, and $p$ such that $mnp>1$,
one can perform $MM(K)$ by using $cK^{\omega}$
arithmetic operations for any $K$,
 $\omega=\omega_{m,n,p,r}=3\log_{mkn}(r)$,
and a constant $c$ independent of $K$.
\end{theorem}

For further applications of the duality technique, see 
efficient bilinear algorithms
for FIR-filters and multiplication of complex numbers and polynomials
in \cite{W80}.

\section{Trilinear Aggregation}\label{straggr}

Aggregation  technique is well-known in 
business, economics, computer science, telecommunication,
natural sciences, medicine, and statistics. 
The idea is to  mass together or cluster independent but similar units
into much fewer aggregates. Their study is simpler,
but its results are supposed to
characterize all these units
either directly or by means of special 
disaggregation techniques.
Such aggregation/disaggregation processes
 proposed in \cite{MP80} became a basis for creating
the field of {\em Algebraic Multigrid},  now quite popular. 
 
Aggregation/disaggregation techniques are behind
the acceleration of MM in Example~\ref{exw0},
which was preceded by similar application of this
technique to polynomial evaluation with preprocessing of 
coefficients \cite{P66}, \cite{K97}.
In that example one first rewrite $u_{2i}v_{2i}=v_{2i}u_{2i}$ 
by using commutativity of multiplication (which does not hold if 
$u_{2i}$ and $v_{2i}$ are matrices), then
 aggregates the terms 
$u_{2i-1}v_{2i-1}$ and $v_{2i}u_{2i}$
into the single term 
$(u_{2i-1}+v_{2i})(v_{2i-1}+u_{2i})$,
thus saving 50\% of scalar multiplications,
that is, $0.5 n^3$ for $n\times n$ MM.
For disaggregation one subtracts 
the {\em correction terms}
$u_{2i-1}u_{2i}$ and $v_{2i}v_{2i-1}$.
$n\times n$ MM involves  $0.5n^2$ pairs of such products,
and so correction involves just $n^2$ scalar multiplications
overall, which is a small sacrifice compared to
saving $0.5n^3$ scalar multiplications.
 
The papers \cite{P72} and  \cite{P78} strengthen  
aggregation  based on
 restating MM as the problem of the decomposition
of the trilinear form $trace(UVW)$, so that 
some additional
links among the subscripts of the input and output entries
enable stronger aggregation and faster MM.

Various implementations of this technique appeared
 in \cite{P78},   \cite{P79},  \cite{P80}, and a number of  subsequent
papers. For demonstration we apply it to
{\em Disjoint MM} where we compute two independent matrix products $AB$ and $UV$
by decomposing the trilinear form
$trace(XYZ+UVW)=\sum_{i,j,k=1}^{m,n,p}(x_{ij}y_{jk}z_{ki}+u_{jk}v_{ki}w_{ij})$
into the sum of rank-1 tensors.

Let 
$T=\sum_{i,j,k=1}^{m,n,p}(x_{ij}+u_{jk})(y_{jk}+v_{ki})(z_{ki}+w_{ij})$
denote a trilinear aggregate made up of the 2 monomials 
$x_{ij}y_{jk}z_{ki}$ and $u_{jk}v_{ki}w_{ij}$
and let
 $T_1=\sum_{i,j=1}^{m,n}x_{ij}s_{ij}w_{ij}$,
$T_2=\sum_{j,k=1}^{n,p}u_{jk}y_{jk}r_{jk}$, and $T_3=\sum_{k,i=1}^{p,m}q_{ik}v_{ki}z_{ki}$
denote 3 groups of  correction terms,
where $q_{ik}=\sum_{j=1}^{k}(u_{ij}+u_{jk})$,
$s_{ij}=\sum_{k=1}^{n}(y_{jk}+v_{ki})$, and $r_{jk}=\sum_{i=1}^{m}(z_{ki}+w_{ij})$.
Then the equation $trace(XYZ+UVW)=T-T_1-T_2-T_3$ defines
a trilinear decomposition
 of rank $mnp+mn+np+pm$ (rather than the straightforward $2mnp$).

Table~\ref{tabaggr2} displays this aggregation/disaggregation technique.

\begin{table}[ht] 
\caption{Aggregation/disaggregation of a pair of terms.}
\label{tabaggr2}
  \begin{center}
\begin{tabular}{| c | c |c |}
      \hline
 $x_{ij}$ & $y_{jk}$ & $z_{ki}$  \\ \hline
 $u_{jk}$ &  $v_{ki}$   & $w_{ij}$ \\ \hline
 \end{tabular}
\end{center}
\end{table}

The product of the 3 sums  of
pairs on input entries in each of the 3  columns of the  table 
is an aggregate.
The 2 products of triples of entries of each of the 2 rows
are the output terms 
$x_{ij}y_{jk}z_{ki}$
and $u_{jk}v_{ki}w_{ij}$.
The cross-products of other triples 
of the table define 6 correction terms.
Their sum over all $n^3$ triples of indices 
$i,j$ and $k$ has rank $2(mn+np+pm)$.
By  
subtracting this sum from the
sum of all $mnp$ aggregates,
we decompose $2mnp$ terms of $trace(XYZ+UVW)$
 into the sum of $mnp+2(mn+np+pm)$
terms. 
For $m=n=p=34$ this implies a decomposition of 
rank $n^3+6n^2$ for a pair of disjoint $MM(n)$.
 
Demonstration of 
the power  of trilinear aggregation
can be made most transparent
 for Disjoint MM, whose natural link to 
trilinear aggregation 
has  been shown in \cite{P84},  \cite[Section 5]{P84a},
 \cite[Section 12]{P84b}, and \cite{LPS92}.
Such constructions 
 for Disjoint MM, however,
can frequently be  extended to $MM(n)$. 
In particular, by playing with odd and even subscripts
of the matrix entries,
the paper \cite{P72} obtained
a trilinear decomposition of rank $0.5 n^3+3n^2$ for $MM(n)$ and any even $n$
by means of 
extending 
the above decomposition of $trace(XYZ+UVW)$.
This implied the $MM$ exponent $\log_n (0.5n^3+3n^2)$,
which is less than 2.85 for  $n=34$. 
 
The paper \cite{P78} defined
a trilinear decomposition and bilinear algorithms
 of rank $(n^3-4n)/3+6n^2$ 
for  $MM(n)$, $n=2s$, and any positive integer $s$.
Substitute $n=70$ and obtain the MM exponent  2.7962.
Then again it is convenient to demonstrate 
this design for Disjoint MM 
associated with a decomposition 
of the trilinear form $trace(XYZ+UVW+ABC)$.
The basic step is the  
aggregation/disaggregation defined by Table~\ref{tabaggr3}.

\begin{table}[ht] 
\caption{Aggregation/disaggregation of a triple of terms.}
\label{tabaggr3}
  \begin{center}
\begin{tabular}{| c | c |c |}
      \hline
 $x_{ij}$ & $y_{jk}$ & $z_{ki}$  \\ \hline
 $u_{jk}$ &  $v_{ki}$   & $w_{ij}$ \\ \hline
 $a_{ki}$   & $b_{ij}$ & $c_{jk}$ \\ \hline
 \end{tabular}
\end{center}
\end{table}

Sum the $mkn$
aggregates 
$(x_{ij}+u_{jk}+a_{ki})(y_{jk}+v_{ki}+b_{ij})(z_{ki}+w_{ij}+c_{jk})$,
subtract order of $n^2$ correction terms,  and
obtain a decomposition of rank $n^3+O(n^2)$
for $trace(XYZ+UVW+ABC)$,
versus the straightforward $3n^3$.
The trace represents 
3 disjoint problems of $MM(n)$, that is,
the computation of
the 3 independent matrix products $XY$, $UV$,  and $AB$
of size  $n\times n$ (and can be readily extended to 
 3 MM products of sizes $m\times n$
by $n\times p$, $n\times p$
by $p\times m$, and $p\times m$
by $m\times n$),
and we obtain a bilinear algorithm of rank $n^3+O(n^2)$
for this bilinear task.

With a little more work 
one obtains a similar 
trilinear decomposition
 of rank $(n^3-4n)/3+6n^2$ 
for  $MM(n)$, $n=2s$, and any positive integer $s$
(see  \cite{P78}).
For $n=70$ we arrive at an upper bound 2.7962 on the MM exponent.
By  refining this construction the algorithm of  \cite{P82}
decreased the upper bound below 2.7734.

~~

\section {APA Algorithms and  Bini's Theorem}\label{sapa}

The 
technique of
{\em Any Precision Approximation} (hereafter we use the acronym 
{\em APA}) was another basic ingredient of the algorithms 
supporting the decrease of the exponent of MM. The paper \cite{BCLR79}
achieved the first APA acceleration
of MM, by yielding the exponent 2.7799.
According to \cite{R79},  this came from 
 computer search for  
partial $MM(2)$ where the goal was the computation of only 3 entries 
of  $2 \times 2$ matrix product.

Next we demonstrate the combination of APA and trilinear aggregation,
which is more transparent and immediately produces the exponent 2.66.
Consider the following table.

\begin{table}[ht] 
\caption{APA aggregation/disaggregation of a pair of terms.}
\label{tabaggr2apa}
  \begin{center}
\begin{tabular}{| c | c |c |}
      \hline
 $x_{ij}$ & $y_{jk}$ & $\lambda^2 z_{ki}$  \\ \hline
 $\lambda u_{jk}$ &  $\lambda v_{ki}$   & $w_{ij}$ \\ \hline
 \end{tabular}
\end{center}
\end{table}

It defines the aggregate
$(x_{ij}+\lambda u_{jk})(y_{jk}+
\lambda v_{ki})(\lambda^2 z_{ki}+w_{ij})$
and 3 correction terms, similarly to Table~\ref{tabaggr2},
but with a decisive difference -- the term
$\lambda^3 (x_{ij}+u_{jk})v_{ki}z_{ki}$
has a smaller order of magnitude as $\lambda\rightarrow 0$.
Therefore we arrive at  trilinear decomposition
$$trace(XYZ+UVW)=T-T_1-T_2+O(\lambda)$$
where
$$T=\lambda^{-1} \sum_{i,j,k=1}^{m,k,n}(x_{ij}+\lambda u_{jk})(y_{jk}+
\lambda v_{ki})(\lambda^2 z_{ki}+w_{ij}),~
 T_1=\sum_{i,j=1}^{m,k}x_{ij}s_{ij}w_{ij},~  
T_2=\sum_{j,k=1}^{k,n}u_{jk}y_{jk}r_{jk},$$
$s_{ij}=\sum_{k=1}^{n}(y_{jk}+\lambda v_{ki})$,
 and $r_{jk}=\sum_{i=1}^{m}(\lambda^2 z_{ki}+w_{ij})$.

Drop the terms of order $\lambda$,
and obtain 
a decomposition for Disjoint $MM(m,n,p)$ having
  {\em  border}  rank  
 $mnp+mn+np$. For $m=p=7$, $n=1$
this implies an {\em APA exponent} of MM 
 $\omega=3\log_{49}31.5<2.66$.
(Here we use Sch{\"o}nhage's result of \cite{S81}
that deduces the MM exponent from Disjoint MM.)

The above APA decomposition using  $mnp+mn+np$ terms
is numerically unstable. Indeed
we would corrupt the output
if we drop the summands 
$\lambda u_{jk}$ and $\lambda v_{ki}$
in the sums $x_{ij}+\lambda u_{jk}$
and $y_{jk}+\lambda v_{ki}$ in
 the aggregate
$(x_{ij}+\lambda u_{jk})(y_{jk}+
\lambda v_{ki})(\lambda^2 z_{ki}+w_{ij})$,
but keeping these summands
doubles the  precision required for the
representation of these sums.
Similarly all other known bilinear APA algorithms 
are prone to numerical stability problems
if their  rank exceeds 
their border rank.

 In \cite{B80}, however, Bini  proved
 that,
for the purpose of decreasing the exponent of  MM,
this deficiency is immaterial if we allow 
unrestricted recursive processes, that is, 
{\em if we ignore the curse of recursion.}
Namely he proved that Theorem~\ref{thdual} 
holds even if border rank 
replaces rank in its statement.
Bini proved this result for MM,
but Sch{\"o}nhage in \cite{S81}
extended it to Disjoint MM.
Both proofs yield acceleration of straightforward MM only 
where its 
input size becomes huge because of 
 the curse of recursion.
\begin{theorem}\label{thdualapa}  
We can perform $MM(K)$ by using $\bar cK^{\omega}$
arithmetic operations for any $K$, where
$\bar c$ is  a constant independent of $K$ and
 $\omega=\omega_{m,n,p,r}=3\log_{mkn}(r)$,
provided that we are given
a bilinear or trilinear APA algorithm having a border rank $r$
for  $MM(m,n,p)$ and
 a 4-tuple of integers $r$, $m$, $n$, and $p$ such that $mnp>1$.
\end{theorem}

\begin{proof}
First observe that by means of
interpolation we can extend any  APA 
algorithm of border rank $r$
using a polynomial $q(\lambda)$ in $\lambda$,
 of a degree $d$  
 to a $\lambda$-free bilinear algorithm 
for $MM(m,n,p)$ of rank
$(2d+1)r$. 

Now apply such an APA algorithm recursively.
Notice that every recursive step 
squares the original problem size $mnp$ but only doubles the degree of $\lambda$. 
After $h$ steps, we arrive at an APA algorithm of degree $2^hd$
for the MM problem of size $(mnp)^{2^h}$,
and then 
(cf. Theorem~\ref{thdual})
the interpolation factor $2(2^hd+1)$ only implies an increase
of the MM exponent $\omega$ by a tiny factor, which converges to 1 as the number of 
recursive steps grows to the infinity. 
\end{proof}

\section{Inner Product Computation  and
Summation by Means of  APA Techniques
}\label{sbsg}

Next we cover an application of APA techniques beyond MM.

Recall the APA algorithm of Section~\ref{sapa},
let the entries $x_{ij}$, $y_{jk}$, $u_{jk}$, and $v_{ki}$
be integers in the range $[0,2^d)$, and choose 
$\lambda=2^d$. Notice that the 
product $(x_{ij}+\lambda y_{jk})(u_{jk}+\lambda v_{ki})$
fits the length $L$ of the computer word if $L\ge 4d$.
Moreover if the ratio $L/d$ is large enough, we can perform the
APA computations of Section~\ref{sapa} within the
precision $L$. 
\cite[Section 40]{P84b} exploits such observations further 
and devise efficient algorithms for multiplication of
vectors and matrices filled with bounded integers.
Next we recall that technique and  in Example~\ref{exsum}
show its surprising extension.

Suppose that the coefficient vector 
of a polynomial $v(\lambda)=\sum_{i=0}^{n-1} v_i\lambda^{i}$
is filled with integers from the  semi-open segment $[0,2^d)$
of the real axis for a positive integer $d$.
Represent this vector
by the $2^d$-ary integer
$v(2^d)=\sum_{i=0}^{n-1} v_i2^{di}$. 
Generally the interpolation 
to a polynomial of degree $n-1$ requires its evaluation at $n$ knots,
but in the above special case we  only need the evaluation at the single knot $2^d$.
Now suppose that  all coefficients $v_i$ are integers from the semi-open
segment 
$[q,r)$ for any pair of integers $q$ and $r$, $q<r$.
Then we can apply the above recipe to compute the shifted vector 
${\bf u}=(u_i)_{i=0}^{n-1}=(v_i-q)_{i=0}^{n-1}$, having all its components 
in the  semi-open segment $[0,s)$
for $s=r-q$. We can finally recover the vector ${\bf v}$ from ${\bf u}$. 
By following \cite{P80a} and \cite{BP86},
we call this technique 
{\em binary segmentation}. Its history
can be traced back to  \cite{FP74},
and one can even view it as 
an application of the Kronecker map,
although
having specific computational flavor. 

Next we follow \cite[Example 40.3, pages 196--197]{P84b} to compute
the inner product of two integer vectors,
then 
extend the algorithm to
summation, and finally 
list  various other applications of
binary segmentation.

\begin{example}\label{exinner}  ({\rm The inner product of two integer vectors}, cf.\ \cite[Example 40.3]{P84b}.)
Assume two nonnegative integers $g$ and $h$ and two vectors ${\bf u}=(u_i)_{i=0}^{n-1}$
and  ${\bf v}=(v_i)_{i=0}^{n-1}$ with nonnegative integer coordinates
 in two semi-open segments, namely, $[0,2^g)$ for the 
coordinates of ${\bf u}$
and $[0,2^h)$ for the coordinates of ${\bf v}$.
The straightforward algorithm 
for the inner product ${\bf u}^T{\bf v}=\sum_{i=0}^{n-1} u_iv_i$
  first computes the $n$ products
$u_iv_i$ for $i=0,1,\dots,n-1$
and then sums them. This involves $n$ multiplications and $n-1$ additions.
Instead, however, we can just
 multiply a pair of bounded nonnegative integers, 
apply binary segmentation to the product, and output the 
desired inner product. 
Namely, introduce the two polynomials $u(x)=\sum_{i=0}^{n-1} u_ix^i$
and  $v(x)=\sum_{i=0}^{n-1} v_ix^{n-1-i}$. Their 
product is the polynomial $q(x)=u(x)v(x)=\sum_{i=0}^{2n-2} q_ix^i$ 
with integer coefficients in the segment $[0,2^k)$
for $k=g+h+\lceil \log_2 n\rceil$.
The coefficient
$q_{n-1}=\sum_{i=0}^{n-1} u_iv_i$
is precisely the inner product ${\bf u}^T{\bf v}$.
Represent the polynomials $u(x)$ and $v(x)$ by their 
integer values $u(2^k)$ and $v(2^k)$
at the point $2^k$. Clearly, they lie in the semi-open segments 
$r_u=[0,2^{nk+g})$ and $r_v=[0,2^{nk+h})$, 
respectively.
Now compute the integer $q(2^k)=u(2^k)v(2^k)$,
lying in the segment  
$[0,2^{2nk+g+h})$, and 
recover the coefficient $q_{n-1}={\bf u}^T{\bf v}$
by applying binary  segmentation.
\end{example}

\begin{remark}\label{reprod}
We only seek the  coefficient $q_{n-1}$ of the median term 
$q_{n-1}x^{n-1}$
of the polynomial
$u(x)v(x)$. This term lies in  the segment
$[2^{(n-1)k},2^{(n-1)k+g+h})$, and 
the next challenge
is to 
optimize 
its computation. 
Is such a more narrow task substantially simpler 
than the multiplication of two integers lying in the segments
$r_u$ and $r_v$?
\end{remark}

\begin{example}\label{exsum} ({\rm Summation of bounded integers.})
For ${\bf u}=(1)_{i=0}^{n-1}$, $g=0$, and $k=h+\lceil \log_2 n\rceil$ or 
${\bf v}=(1)_{i=0}^{n-1}$, $h=0$, and $k=g+\lceil \log_2 n\rceil$,
the algorithm of Example~\ref{exinner} outputs the sum of $n$ 
integers. 
\end{example}

\begin{remark}\label{relax}  
In the same way as for polynomial interpolation
in the beginning of this section, we can relax the 
assumption of Examples~\ref{exinner} and~\ref{exsum}
that the input integers are nonnegative.
Moreover, the summation of integers can be extended to
the fundamental problem of the 
summation of binary numbers truncated to a fixed precision.
\end{remark}

In Examples~\ref{exinner} and~\ref{exsum},
 multiplication of
 two long integers followed by binary segmentation replaces
 either $2n-1$ or $n$ arithmetic operations,
respectively.
This increases the
 Boolean (bit-wise operation) cost 
by a factor depending on the Boolean cost 
of computing the
 product of 2 integers or, in view of
Remark~\ref{reprod}, of computing the  
median segment in the binary representation
of the product.
The increase is minor if we multiply integers in 
nearly linear Boolean time (see the supporting algorithms for such multiplication in
  \cite{SS71}, \cite[Section 7.5]{AHU74}, \cite{F09}),
but grows if we multiply integers
by applying the straightforward algorithm, 
which uses quadratic Boolean time.
Nonetheless, in both cases
 one could still benefit from using Example~\ref{exsum}
if
the necessary bits of the output integer fit the computer word
(i.e. the bits of the middle coefficient are not part of the overflow of the product),
as long as the representation of the vector as an integer
requires no additional cost.
If
  the output
integer  does not fit the 
word length, 
we can apply the same algorithms to the subproblems of smaller sizes,
e.g., we can apply the algorithms of Examples~\ref{exinner}
and~\ref{exsum} to
compute the inner products of some subvectors
 or partial sums of integers, respectively.

Other applications of binary segmentation include
 polynomial multiplication (that is, the computation of  the 
convolution of vectors)   
 \cite{FP74}, \cite{S82},  some basic linear algebra
 computations \cite[Examples 40.1--40.3]{P84b}, 
  polynomial division  \cite{BP86}, \cite{S82},
 computing polynomial GCD   \cite{CGG84},
and  discrete Fourier transform 
 \cite{S82}. Binary segmentation can be potentially efficient in  
computations with Boolean vectors and matrices.
E.g., recall that Boolean MM is reduced to MM whose input and output entries 
are some bounded nonnegative integers (see  \cite[Proof of Theorem 6.9]{AHU74}). 
Quantized  
tensor decompositions is another promising application area (cf.\   
 \cite{T03}, \cite{O09},
\cite{O10}, \cite{K11}, \cite{OT11}, \cite{GKT13}).

\section {Summary of the Study of the MM Exponents after 1978}\label{sfrth}

In 1979--81 and then again in 1986
the exponent of infeasible MM was significantly decreased
based on combination of trilinear aggregation, Disjoint MM, and APA techniques with
unrestricted use of recursion.

All  supporting algorithms have been
built on the top of  the techniques of 
the preceding papers. 

More and more lenient basic bilinear/trilinear decompositions 
of small rank were chosen for  
 Disjoint MM of small sizes and  since 1986 for small
bilinear/trilinear 
problems  similar to Disjoint MM.
Transition back to MM 
relied on   
 nested recursion,
 consistently
intensified;
 consequently  
acceleration of 
the straightforward algorithm began only with 
MM of astronomical sizes.

By 1987 the power of these techniques seems to be exhausted,
and then the progress has  stopped until 2010. Since then it is
 moving from the bound 2.376 of \cite{CW90} towards 2.37 with the snail's speed.

That direction was prompted by the cited results of the seminal papers
\cite{B80} by Dario Bini and  \cite{S81}  by  
Arnold Sch{\"o}nhage.
Sch{\"o}nhage, however, has concluded the introduction of \cite{S81} with
 pointing out 
that all new exponents of MM were just "of theoretical interest"
because they were valid only 
for the inputs "beyond any practical size" and that 
"Pan's estimates of 1978 for moderate" input sizes  
were "still unbeaten".  Actually, as we can see in Figure~\ref{fig1}, 
the exponent 2.7962 of 1978 for  
$MM(n)$ restricted to $n\le 1,000,000$  has been
 successively  decreased in \cite{P79},
 \cite{P80},  \cite{P81}, and  \cite{P82}  (cf.\ also \cite{P84}), 
although by  
small margins. 
As of December 2016, the exponent 2.7734 of  \cite{P82} 
is still record low for $MM(n)$ with $n\le 1,000,000$. 

 Figures
\ref{fig1} and~\ref{fig2} display
chronological  decrease of the  exponents of
$MM(n)$ for $n\le 1,000,000$ 
and for unrestricted $n$, 
respectively. 
The 
supporting algorithms of Figure
\ref{fig1} rely solely on trilinear aggregation,
and the associated overhead constants 
 are small.
Some of these algorithms have been refined in 
\cite{P84}, \cite{LPS92} and  
implemented
  in \cite{K99} and \cite{K04}.
 
All other algorithms of
Figure~\ref{fig2} (not supporting Figure~\ref{fig1})
employ trilinear aggregation as well
(cf.  \cite[page 255]{CW90}),
but also employ other techniques and
suffer from the curse of recursion.

The figures link each exponent to its recorded
publication in a journal, a conference proceedings, or
as a research report.
As we already mentioned,  the MM exponent of \cite{P78}
significantly decreased 
  in 1979--1981 and 1986.  
It has been 
updated at least 4 times during the single year of 1979:
reaching below the values 2.801 in February in  Research Report 
 \cite{P80}; 2.7799  
(as an APA MM exponent)
in \cite{BCLR79} in June
and (as an MM  exponent)
in \cite{B80};
2.548 in \cite{S81}, and
2.522 in \cite{P81}. Both of the latter two exponents appeared in the   
book of abstracts 
of the conference on the Computational Complexity in Oberwolfach, West Germany, 
organized by Schnorr, Sch{\"o}nhage and Strassen in October
(cf.\ \cite[page 199]{P84} and \cite{S81}). 
The exponent 2.496 of \cite{CW82}
 was reported in October 1981 at the IEEE FOCS'81,
and in  Figure
\ref{fig1} we place it after the exponent 2.517
of the paper \cite{R82} of 1982,
which was submitted in March 1980.
The Research Report version of the paper \cite{CW90}
appeared in August of 1986, 
but in  Figure
\ref{fig2} we place  \cite{CW90} after the paper  
\cite{S86},  published in October of 1986 in the Proceedings of the IEEE FOCS,
because  the paper \cite{S86} has been  
submitted to FOCS'86 in the Spring of 1986 and has been
widely circulated afterwards.
One could complete the historical account
of  Figure
\ref{fig2}
by including the  
exponents
 2.7804 (announced in
the Fall of 1978 in \cite{P78}
 and
 superseded in February 1979 when the paper was
submitted \cite{P80})
and 
2.5218007, which
 decreased the exponent 2.5218128 of \cite{P79}  
and appeared at the end of
the final version of 
 \cite{S81}  in 1981,
that is, before the publication, but after the submission of the 
exponent 2.517 of \cite{R82}. 

We refer the reader 
to  \cite{C82}, \cite{LR83},  \cite{C97},  \cite{HP98}, \cite{KZHP08}, 
  \cite{LG12}, and the references therein
for similar progress in asymptotic acceleration of rectangular MM.

\begin{figure}[ht]
	\caption{$MM(n)$ exponents for $n\le 1,000,000$.}
	\label{fig1}
		\centering
\includegraphics[height=3in,keepaspectratio]{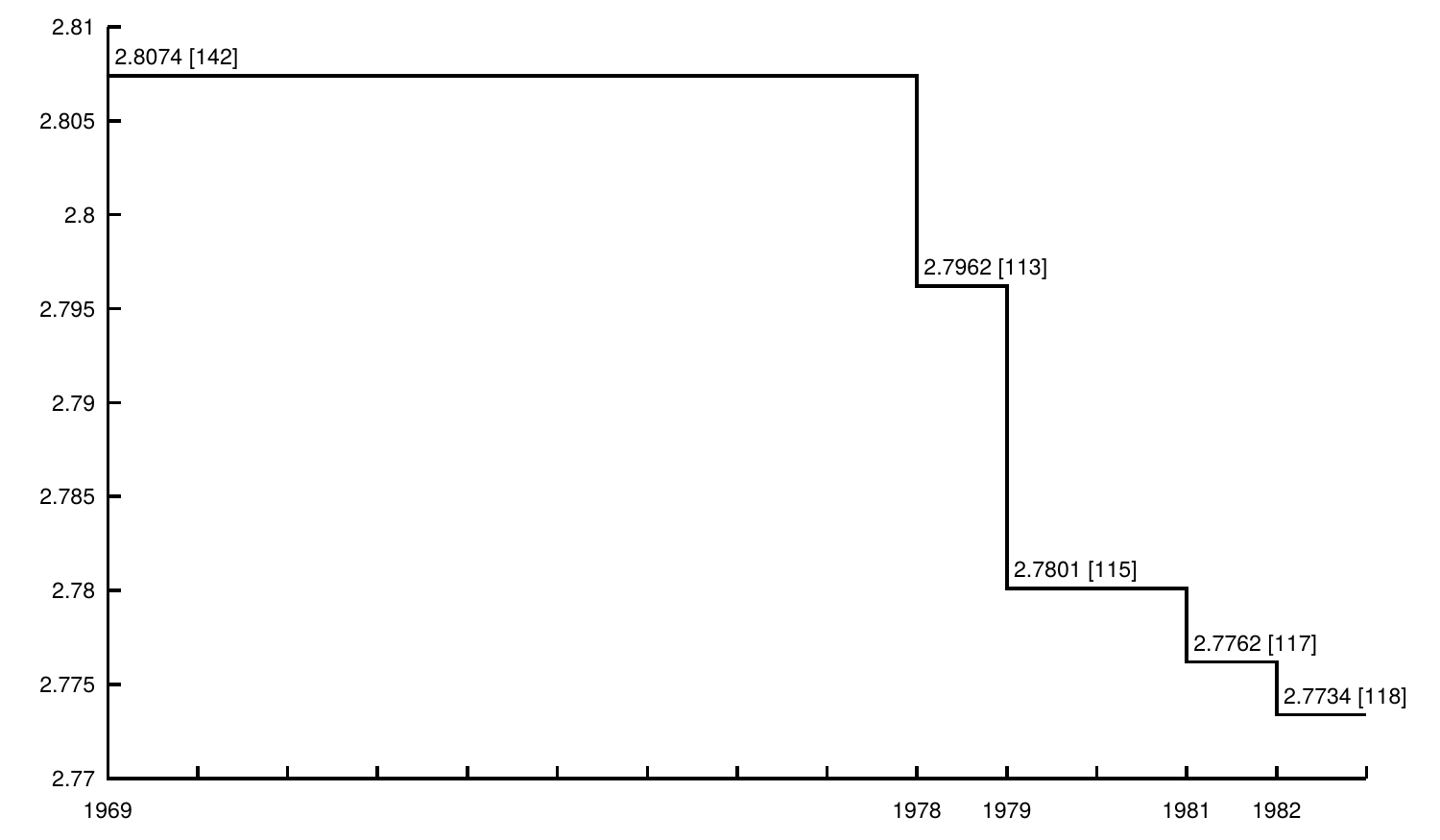}
\end{figure}

\begin{figure}[tbp]  
  \centering
  \caption{$MM(n)$ exponents for unrestricted $n$.}
\includegraphics[height=3in,keepaspectratio]{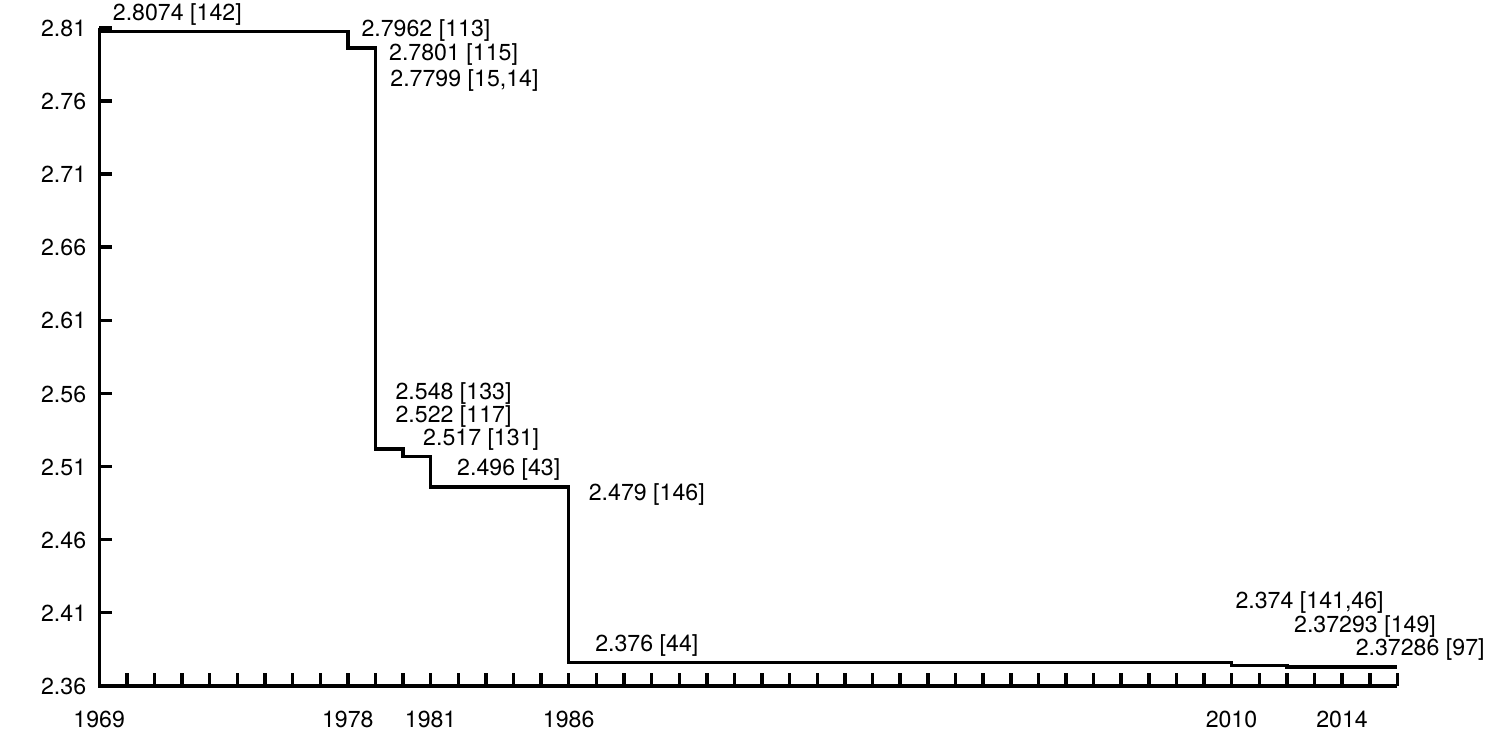}
  \label{fig2}
\end{figure}

\medskip

\section{Applications of Fast MM, Its Links to Other Subject Areas,  Impacts of Its Study,
 and Implementation (Briefly)}\label{sapplcs}

The decrease of the exponent of MM
 implies theoretical acceleration of the solution  of a number of important 
 problems in various areas of computations in Algebra and Computer Science,
such as Boolean MM, computation of paths and distances in graphs,
parsing context-free grammars,
the solution of a nonsingular 
linear system of  equations, 
 computations of the inverse,  determinant, characteristic and minimal 
polynomials, and
various factorizations  
of a matrix. 

 See  \cite{S69}, \cite{BH74}, 
 \cite[Sections 6.3--6.6]{AHU74}, \cite[pages 49--51]{BM75},
 \cite[Chapter 2]{BP94},  \cite{AGM97},  \cite{HP98},  \cite{DI00},
 \cite{L02},  \cite{Z02},
 \cite{YZ04}, \cite{YZ05}, \cite{YZ05a},  \cite{KSV06},
  \cite{BJS08}, \cite{KZHP08},
 \cite{AP09}, \cite{DP09}, \cite{Y09}, \cite{SM10}, 
 \cite{LG12},
 \cite{ASU13},
 \cite{S15}, \cite{SY15},  \cite{N16}, \cite{NS16}, \cite{NS16a}, \cite{P16},
and the bibliography therein and  notice that some new important
applications have been found very recently, e.g., in 4 papers at ISSAC 2016.

The reduction of other computational problems to MM increases
the overhead, which is already immense for the algorithms 
supporting the record exponent of MM. Such a reduction, however, 
can still be valuable because it reveals some links 
of independent interest and 
because  by applying fast  algorithms or even the straightforward algorithm for
feasible MM
one can take advantage of using block matrix software and parallel acceleration.

Research on fast feasible MM had a variety of
 feasible links to other subject areas.
The work on bilinear and trilinear decompositions for fast MM was an
important part of the study of such decompositions in the general area of Algebraic Computations and
has led to new insights and 
 techniques.

\begin{itemize}
\item
We have already cited historical importance of the demonstration in 1972 in~\cite{P72}
 of the power
of tensor decompositions and 
valuable applications of 
duality to the design of efficient bilinear algorithms. 
\item
Trilinear aggregation
was also a surprising demonstration of the power of ag\-gre\-ga\-tion/dis\-ag\-gre\-ga\-tion methods.
\item
 More applications of this kind can follow in the future, such as 
a quite unexpected application of APA techniques to
the computation of inner products and
summation presented in Section~\ref{sbsg}.
\item
It may be surprising, but apart from trilinear aggregation and the
APA method, the advanced and amazing techniques 
developed for decreasing the exponent of infeasible MM had no theoretical (as well as practical) applications 
to other areas of computations or algebra (and have made no impacts on actual work on MM as well).
\end{itemize}
 
 Of course, such impacts on the practice of performing this fundamental operation of modern computations 
are the main motivation and goal of the study of fast MM, and indeed
 recursive bilinear algorithms based on $2\times 2$ 
Strassen's and particularly  Winograd's brilliant designs of
Examples~\ref{ex0}  and~\ref{ex1}, respectively,
are a valuable part of modern software for MM.
 
In Section \ref{sfcs} we commented on numerical stability issues for fast 
feasible MM,
and we refer the reader to
 \cite{B88}, \cite{H90},
\cite{DHSS94}, \cite{DGP04}, \cite{DGP08}, \cite{BDPZ09},
\cite{DN09},
 \cite[Chapter 1]{GL13}, \cite{BB15}, \cite{BBDLS15}, \cite{BDKSa},
and the bibliography therein for their previous and current numerical implementation.
In the next section we discuss symbolic application of the latter  algorithm
(WRB-MM) in some detail.

It is very encouraging  to observe dramatic increase of the activity 
in  numerical and symbolic implementation of fast MM in recent years,
towards decreasing communication cost and improving parallel implementation, 
  in good accordance with the decrease of the arithmetic cost.

\section{Fast methods in exact computational linear algebra}\label{sec:exactcomp}

Next we discuss 
applications and implementations of
\emph{exact} fast matrix multiplication (MM). 
As we mentioned in Section \ref{ssbj}, we first review how most of the 
exact
linear algebra can be reduced to MM over small
finite fields. Then we highlight the differences in the design of
approximate and exact implementations of MM taking into account nowadays
processor and memory hierarchies.

\subsection{Acceleration of computations via reductions to MM}

The design of matrix multiplication routines over a word size finite
field is the main building block for computations in exact dense
linear algebra, represented
with  
coefficient domains of two kinds: word-size discrete entries (mainly in finite
fields of small cardinality) and variable size entries (larger
finite fields, integers, or polynomials). 
Indeed, efficient methods for the latter task usually reduce it to
the former one: reductions are made by means of
evaluation/interpolation (Chinese remaindering over the integers) or
 by means of lifting small size solutions (via Hensel-like or high-order
lifting~\cite{Storjohann:2005:HighOrder}), see,
e.g.,~\cite{Kaltofen:2011:EACM,Dumas:2012:HFF,pernet:tel-01094212} for
recent surveys. 

Over word-size finite fields, efficient implementations of MM have
been obtained by means of effective  
algorithmic reductions
originated in the complexity analysis.
In this case  reductions have been obtained directly to MM. 
Already Strassen in 1969  extended MM 
to matrix inversion (we sketch this in Examples~\ref{exTRSM} and~\ref{exLU}
below), then  Bunch and Hopcroft~\cite{BH74}  extended 
 MM to invertible LUP factorization, and further extensions followed in the
eighties, for instance, in~\cite{Ibarra:1982:LSP}
 to the computation of 
the matrix rank and of Gaussian
elimination. Efficient algorithms
whose complexity is sensitive to the rank or to the rank profile have
only very recently been
discovered~\cite{Jeannerod201346,DPS16,SY15}. 
One of the latest
reductions to MM was that of the characteristic polynomial, which
has extended the
seminal work of~\cite{Keller-Gehrig:1985:FAC} 
more than thirty years afterwards~\cite{Pernet:2007:FAC}.

\begin{example}\label{exTRSM} {\rm Triangular system solving by means of
  reduction to matrix multiplication.}\\
Denote by $X=TRSM(U,B)$ the solution  to the linear matrix equation $UX=B$
with a matrix $B$
on the right-hand side and an upper triangular invertible matrix $U$. 
Recursively cut the matrices $U$, $B$ and $X$ in halves as follows: 
$U=\begin{bmatrix} U_1 & V\\ & U_2\end{bmatrix}$, 
$X=\begin{bmatrix} X_1 \\ X_2\end{bmatrix}$, and 
$B= \begin{bmatrix} B_1 \\ B_2\end{bmatrix}$;
then obtain an efficient reduction of the solution to MM by means of the following algorithm:
\begin{enumerate}
\item Recursively compute $X_2=TRSM(U_2,B_2)$;
\item Compute $B'_1=B_1-VX_2$; // via fast MM
\item Recursively compute  $X_1=TRSM(U_1,B'_1)$.
\end{enumerate}
The only operations performed are fast MMs.
Asymptotically the low complexity is preserved, and in practice
this reduction can be made very efficient, even in the case of exact 
computations,
where intermediate reductions might occur, as shown, e.g.,
in~\cite[\S~4]{DGP08}.
\end{example}

\begin{example}\label{exLU} {\rm Reduction of  LU factorization to MM.}\\
For simplicity, consider an invertible matrix $A$ having generic rank profile (that is, having \emph{all
 its leading principal minors also invertible}).
Recursively cut $A$ into halves in both dimensions,  that is, represent it as $2\times 2$ block matrix,
$A=\begin{bmatrix} A_1 & A_2 \\ A_3 & A_4 \end{bmatrix}$.
Then an efficient triangularization $A=LU$ can be computed by means of the
following algorithm:
\begin{enumerate}
\item Recursively  compute $L_1 U_1 = A_1$; // Triangularization of the upper
  left block
\item  Compute $G=TRSM(U_1,A_3)$; // $G$ is such that $GU_1=A_3$
\item  Compute $H=TRSM(L_1,A_2)$; // $H$ is such that $L_1H=A_2$
\item  Compute $Z=A_4-GH$; // via fast MM
\item Recursively  compute $L_2,U_2=Z$;
\item Return $L=\begin{bmatrix} L_1 &\\ G & L_2\end{bmatrix}$ and
  $U=\begin{bmatrix} U_1 & H\\ & U_2\end{bmatrix}$.
\end{enumerate}
Once again, the only operations performed in this algorithm  are MMs,
making it  efficient as long as MM is efficient. This reduction to MMs
would remain efficient in the
extension 
to  the more general
case where rank deficiencies are allowed and pivoting is applied~\cite{DPS16}.
\end{example}

\begin{example}\label{exDixon} {\rm Sketch of linear system solving in arbitrary
  precision.}\\ 
Next we  accelerate the  solution of a linear system of equations 
with arbitrary precision by using the two previous reductions to
 exact MM. 
 The idea is to solve the linear system 
modulo a small prime $p$, by intensively using fast exact MM, and then to reuse this
factorization in Hensel-like $p$-adic lifting producing iterative refinement of the
solution modulo $p^k$.\footnote{Hensel-like $p$-adic lifting is a major tool of computer algebra, which has striking similarity with
the classical  algorithm of iterative refinement in numerical linear algebra.} 
This leads us to the following algorithm:

\medskip

Successively compute
\begin{enumerate}
\item  $L_p,U_p \equiv A \mod p$; // Triangularization modulo a small prime
\item   $x_0 \equiv TRSM(U_p, TRSM(L_p, b) \mod p) \mod p$;
\item $b_1 = \frac{b-Ax_0}{p}$; // this computation is over \ZZ
\item $x_1 \equiv TRSM(U_p, TRSM(L_p, b_1) \mod p) \mod p$;
\item $b_2 = \frac{b_1-Ax_1}{p}$; //  this computation is over \ZZ
\item $x_2 \equiv TRSM(U_p, TRSM(L_p, b_2) \mod p) \mod p$; 
\item $\ldots$
\end{enumerate}
One can easily show that 
$b \equiv A( \sum_{i=0}^{k-1} x_i p^i ) \mod p^k$.

Now recall that  we can recover unique rational number $r=\frac{a}{b}$
from its $p$-adic series representation
truncated at $p^k$ as soon as $k$ is such that $p^k$ is larger than
$2ab$ and the denominator $b$ is coprime to
$p$.\footnote{The recovery of a rational $r$ from its $p$-adic truncated series
  is called \emph{rational number reconstruction} which is an accelerated Euclidean 
  algorithm, see, e.g., 
 \cite[\S~5.10]{GG13} for more details.} Namely, by combining Cramer's rule and Hadamard's bound on determinants,
we represent solution values $x_i$ as rational numbers with bounded denominators
and then recover them from their truncated $p$-adic series representation 
by applying sufficiently many steps of iterative refinement.
 More details on actual values of
$p$ and $k$ can be found in Dixon's seminal paper~\cite{Dixon:1982:Pad}.
We can prove (and this is important for practical application)
that the overhead of this exact iterative refinement is quite a
 small constant. We  show this in
Figure~\ref{fig:Dixon}, for which we used the
LinBox\footnote{\url{https://github.com/linbox-team/linbox}} exact linear
algebra library: despite quite large coefficient growths (for instance, up to
forty thousand bits for the solution to a  $8000\times{}8000$ random
integer matrix with $32$ bits entries), approximate and exact
arbitrary precision times remain essentially proportional. 
\end{example}
\begin{figure}[ht!]
\centering
\caption{Comparison of approximate and arbitrary precision linear system solving
  on an Intel Xeon W3530  @2.80GHz.} 
\label{fig:Dixon}%
\includegraphics[width=\textwidth]{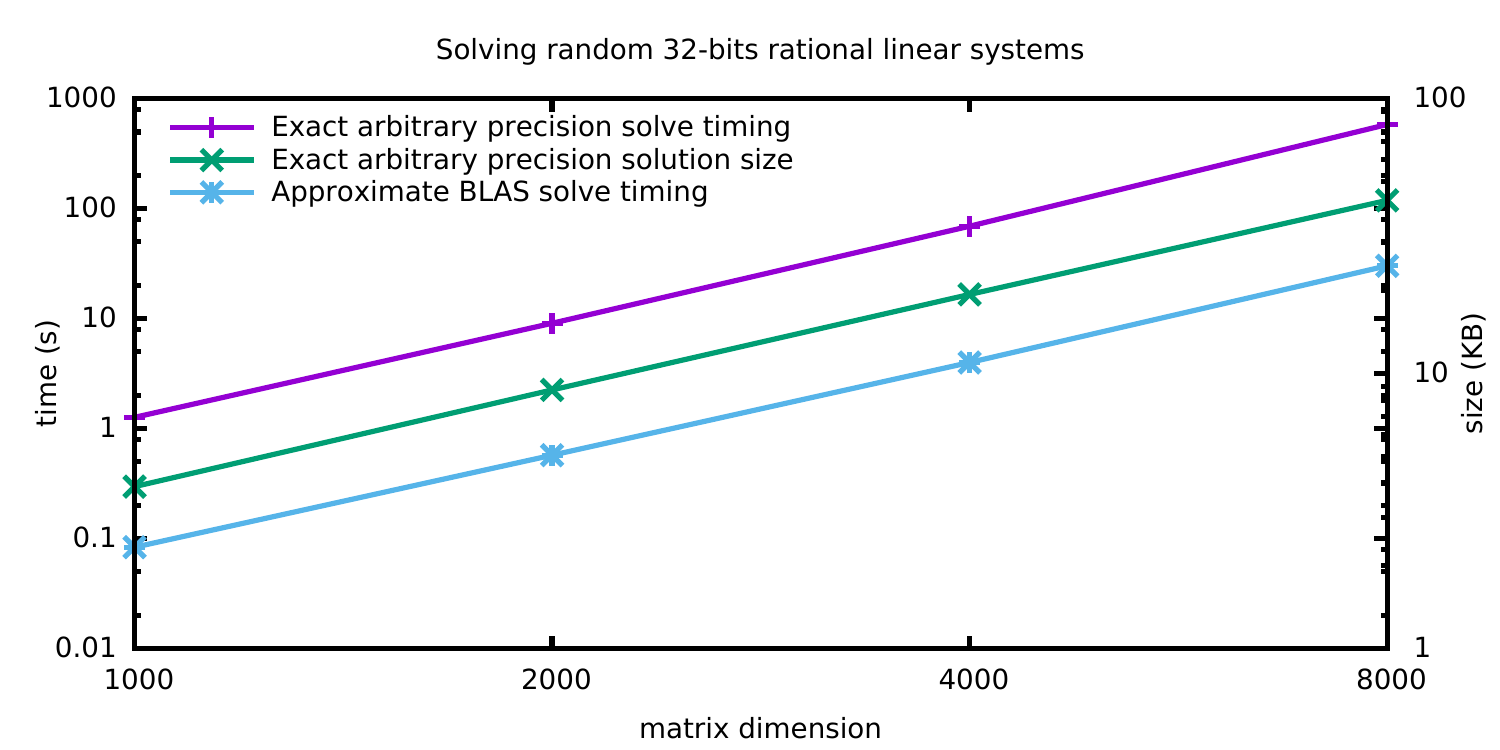}
\end{figure}

We achieve  good practical performance of
computations in linear
algebra  by extensively applying
 reductions to MM, which we mostly perform exactly over small finite fields.

On the one hand, nowadays machine architecture, with their
memory hierarchy, is also well-adapted to the highly homogeneous
structure of the straightforward 
MM.
This is true for
numerical routines, where 
the stability issues have %
been worked out
for fast MM
in~\cite{BL80,H90,DDHK07,BBDLS15}.
This is also true
for exact routines where, 
symmetrically, the
costly reductions (modular or
polynomial) have to be delayed as much as possible.

On the other hand, this straightforward algorithm remains faster in
practice only for small matrices, but the larger memory capacity and the increase
in the number of computing units makes it possible to handle, on desktop
computers, dense matrices  of dimensions in 
several hundred
thousands.
As the practical threshold between straightforward and fast algorithms
is often for matrices of dimensions 
about 500,
several fast routines can be very
effectively combined, yielding the fastest implementations to date, as explained
in the next section.

\subsection{Design of fast exact matrix multiplication over word-size prime
  fields}

The design of matrix multiplication routines over a word size finite
field is the main building block for computations in exact dense
linear algebra.
In practice, a turning point has been the introduction of the
following principles in~\cite{jgd:2002:fflas}:
\begin{enumerate}
\item finite field arithmetic is reduced to integer arithmetic with
  delayed or simultaneous modular reductions;  
\item integer arithmetic is performed by floating point units (in
  order to take advantage of SIMD instructions and of numerical
  routines development -- BLAS);
\item computations are structured in blocks 
in order to optimize the use
  of the memory hierarchy of current architectures;
\item asymptotically fast algorithms are used, mostly
recursive bilinear algorithm  for MM based on  Winograd's
 $2\times 2$ MM of Example~\ref{ex1} (see also~\cite{DHSS94,jgd:2002:fflas,BDPZ09}, hereafter
we  denote this algorithm \strassenwinograd)%
, but also Kaporin's~\cite{K04} and
Bini-Capovani-Lotti-Romani's ~\cite{BCLR79,jgd:2015:bini} algorithms are used.
\end{enumerate}

The idea is to convert a finite field matrix into its integer
representation, then perform the multiplication over the integers
(potentially using, exactly, a floating point representation) and
convert back the result into the finite field. 
First, a floating point representation allows us to use a
sustainable code that rely on more widely supported numerical linear
algebra routines. 
Further, on the one hand, 
 machine division (for reductions)
and integer units are slower than the respective
floating point operations (for
instance, not all
integral operations have yet SIMD support~\cite{HoevenLQ14}).
But, on the other hand, as the computed results are exact,
one can apply  the fastest implementations of
the
asymptotically fast algorithms 
not worrying about
numerical stability problems
(which, even though never serious for fast MM,
as this has been proved
in~\cite{BL80,DDHK07}, can require some
extra work). 

Over finite fields, with arbitrary precision or with polynomials,
it is nowadays much faster to perform additions or even
multiplications than modular or polynomial reductions.
To take the most of the machine representation, reduction
is delayed or grouped using the following techniques:
\begin{definition}\label{defdelayed}
\begin{enumerate}\renewcommand {\theenumi}{(\roman{enumi})}
\item \emph{Delayed reduction} is the action of replacing a classical
  dot product with reductions, $\sum_{i=0}^n Reduce(a_i b_i)$, by an algorithm
  that reduces only from time to time, for instance, when the internal
  representation cannot hold the growth anymore:
  \begin{itemize}
  \item Compute $s=\sum_{i=0}^{k} a_i b_i$; $Reduce(s)$; \\
  \item For $j=1$ to $n/k-1$ compute \\
    \begin{description} \item $s = s + \sum_{i=jk+1}^{(j+1)k} a_i
      b_i$; $Reduce(s)$; \\
    \end{description}
  \item done;
\end{itemize}
\item \emph{Simultaneous modular reduction} is the action of
  performing a single reduction on several coefficients at once:
  in other words, the idea is to replace independent reductions (for instance
  $Reduce(a)$;$Reduce(b)$;$Reduce(c)$) by a single reduction on
  grouped elements (use
  $t=Group(a,b,c)$;$Reduce(t)$;$a,b,c=Ungroup(t)$ instead), see, e.g., \cite{DumFouSal11}
  for more details.
\end{enumerate}
\end{definition}

Moreover, in practice, for exact computations, recursion is important when 
fast algorithms are used: while tiling allows a static placement of blocks
adapted to the memory hierarchy, larger blocks allow 
faster acceleration but also more delays in the application
of reductions (modular reduction, polynomial reduction,
etc.)~\cite{Dumas:2014:parPLUQ}.
Further, the fast algorithms can be used for a few levels of
recursion and then switch back to the, statically placed,
straightforward algorithm on smaller matrices.
Eventually (nowadays, on a current personal computer, this can
typically be between n=500 and n=1000), the exact algorithm sketched
above is based on the numerical BLAS and then applied without any
further recursion. 

There is actually a cascade
of algorithms where, at each recursive cutting, a new choice of the
best suited variant is performed before the final switch~\cite{F74,14BDGPS}.
This cascade is usually non-stationary, i.e., a different scheme, bilinear or
not, can be chosen at every level of recurrence.
For instance, a meta-algorithm starts by partitioning the matrices 
into four blocks followed by application of 
{\strassenwinograd}.
This meta-algorithm is called
on each of the seven multiplications of the blocks; at the recursive
call the meta-algorithm can decide whether to re-apply itself
again to a $2\times 2$ block matrix or to switch either 
to a (2,2,3) APA algorithms (which in turns will also call this meta
algorithm recursively) or to the straightforward algorithm, etc.

This impacts the frequency at which the modular reductions can be
delayed. For instance, with a classical
MM
and 
elements of a prime field modulo $p>2$ represented as integers in
$\{\frac{1-p}{2}\ldots\frac{p-1}{2}\}$, on a type with a mantissa of $m$
bits, the condition is that the modular 
reduction in a scalar product of dimension $k$ can be delayed to the
end as long as $k\left(\frac{p-1}{2}\right)^2<2^m$.

When applying $\ell$ recursive levels of \strassenwinograd 
algorithm, it can be showed instead that
some intermediate computations could grow above this bound
\cite{DGP08}, and the latter condition becomes  
$9^\ell \lfloor\frac{k}{2^\ell}\rfloor\left(\frac{p-1}{2}\right)^2<2^m$. 
 This requires 
to perform by a factor of about $(9/2)^\ell$ more
modular reductions.

Example 
of sequential speed obtained by the 
fgemm routine algorithm of the \linbox library\footnote{Within its
  \fflas module~\cite{DGP08},
  \url{https://github.com/linbox-team/fflas-ffpack}}~\cite{Dumas:2002:icms}
is shown in Table~\ref{tab:fgemm}.

\begin{table}[ht!]
\caption{Effective Gfops ($2n^3/time/10^9$) of matrix
  multiplications: \linbox \fgemm vs \openblas \texttt{(s|d)gemm} on one core of a Xeon E5-4620 @2.20GHz}
\label{tab:fgemm}
\renewcommand{\arraystretch}{1.3}
\center 
\begin{tabular}{lrrrrr}
\hline
$n$& 1024 & 2048 & 4096 & 8192 & 16384\\
\hline
\openblas \sgemm  & {\bf 27.30}& {\bf 28.16}& 28.80 & 29.01 & 29.17\\
$\BigO{n^3}$-\fgemm Mod 37 & 21.90& 24.93 & 26.93& 28.10& 28.62\\
$\BigO{n^{2.81}}$-\fgemm Mod 37 & 22.32& 27.40 & {\bf 32.32}& {\bf 37.75} & {\bf 43.66}\\
\hline
\openblas \dgemm  & 15.31 & 16.01& 16.27 & 16.36 & 16.40\\
$\BigO{n^3}$-\fgemm Mod 131071 & 15.69 & 16.20 & 16.40 & 16.43 & 16.47\\
$\BigO{n^{2.81}}$-\fgemm Mod 131071 & {\bf 16.17}& {\bf 18.05}& {\bf
  20.28}& {\bf 22.87}& {\bf 25.81}\\  
\hline
\end{tabular}
\vspace{5pt}
\end{table}

The efficiency of the \fgemm routine 
is largely due to the
efficiency of the BLAS,
 but as the latter are not using fast
algorithms, exact computations can be faster.
Modulo 37 elements are stored over single precision floats,
 and the \sgemm
subroutine can be used, whereas modulo 131071 elements are stored using double
precision float,
 and the \dgemm subroutine is used. 
Table~\ref{tab:fgemm} first shows that the overhead of performing the modular
reductions in the $\BigO{n^3}$ implementations is very limited if the
matrix is large enough. Then, when enabling
\strassenwinograd $\BigO{n^{2.81}}$ algorithm,  a speed-up of up to $40\%$ can be
attained in both single and double precision arithmetic.
More recently, it has been shown also how
algorithm~\cite{BCLR79} by Bini et al.  could be efficiently put
into practice~\cite{jgd:2015:bini},
also offering some interesting speed-up
for prime fields of size near 10 bits.

\subsection{Memory efficient schedules}
\strassenwinograd algorithm requires external temporary memory
allocations 
in order to store the intermediate linear combinations of blocks. 
With large matrices and current architectures, this can be penalizing
because the memory availability and access to it dominate the overall 
computational costs. 
It is therefore crucial to reduce as much as possible the extra memory
requirements of fast methods. This can be done via a careful
scheduling of the 
steps of \strassenwinograd algorithm:
it is not mandatory to perform 8 pre-additions,  7 
multiplications and  7 post-additions in that order, 
with temporary memory allocations for each of these 22 steps.
Depending on the associated dependency graph, one can choose instead
to reduce the number of allocation by following this graph and overwriting
already allocated memory when the associated variable is not used 
any more. 

For the 
product 
of  $n\times{}n$ matrices, without accumulation, $C\leftarrow{}A \times{} B$,
 \cite{DHSS94} proposed a
schedule requiring, apart from $C$, two extra temporary blocks of size
$\frac{n}{2}\times\frac{n}{2}$ at the first recursive levels, two
extra temporary blocks of size $\frac{n}{4}\times\frac{n}{4}$ for each
of the seven recursive calls, etc. Overall, the needed extra memory is
bounded by $\frac{2}{3}n^2$. 
For the product with accumulation, $C\leftarrow{} C+ A \times{} B$,
for more than ten years the record was three temporaries with an extra
memory bounded by $n^2$~\cite{HussLederman:1996:ISA}, but this was
recently improved to two temporaries 
in~\cite{BDPZ09}. 
Notice that~\cite{BDPZ09} proposed also some schemes
requiring smaller extra memory 
(that can actually be made quite close to zero), with
the same asymptotic complexity as \strassenwinograd algorithm,
although with a larger constant
overhead factor in arithmetic operations. 
Recently M. Bodrato 
proposed a
variant of \strassenwinograd algorithm, which is
symmetric and more suitable to
squaring 
matrices, but 
which uses similar schedules, and therefore keeps
the extra requirements%
~\cite{Bodrato:2010:SMM}.
This is not the case in~\cite{HSHG16} where no extra memory is
required if only one or two recursive levels of fast MM are used, but
at the cost of recomputing many additions.

Finally, in the case of the APA algorithm~\cite{BCLR79} by Bini et al.,
the requirements are also of two temporaries in the product without
accumulation~\cite{jgd:2015:bini}. 
From~\cite{HussLederman:1996:ISA}, new schedules have usually been
discovered by hand, but with the help of a pebble game program, 
which discards rapidly the wrong schedules 
and verifies formally the correct ones.

\subsection{Tiny finite fields}\label{sssec:tiny}
The practical efficiency of 
MM depends 
greatly on
the representation of field elements. 
Thus
we present three kinds of compact representations for 
the elements of a
finite field with
very small cardinality:
 bit-packing 
(for $\F_2$), bit-slicing (say, for 
$\F_3, \F_5, \F_7, \F_{2^3}$, or $\F_{3^2}$), and Kronecker
substitution. These representations are
designed to allow
efficient linear algebra operations, including 
MM:
\begin{definition}\label{deftiny} Compact representations for small
  finite fields.
\begin{enumerate}\renewcommand {\theenumi}{(\roman{enumi})}
\item Over $\F_2$, the method of \emph{the four
    Russians}~\cite{MR0269546}, also called \emph{Greasing}, can be
  used as follows~\cite{AlbrechtBH10}: 
\begin{itemize}
\item A 
64-bit machine word can be used 
in order to represent a row vector of dimension~64.
\item %
Multiplication 
of an $m\times k$ matrix  $A$ by an $k\times n$
matrix $B$ can be done by first storing all $2^k$ $k$-dimensional linear
combinations of rows of $B$ in a table. Then the $i$-th row of the
product is copied from the row of the table indexed by the $i$-th row of $A$.
\item By ordering indices of the table according to a binary Gray Code, each
row of the table can be deduced from the previous one, using only one row
addition. This 
decreases the bit-operation count for
building the table from $k2^kn$ to $2^kn$.
\item Choosing $k=\log_2n$ in this method implies $MM(n)=\BigO{n^3/\log n}$
over~$\F_2$. In practice, the idea 
is once again to use a cascading
algorithm: at first some recursive steps of fast 
MM is performed,
and then, at a size small enough,
 one should
 switch to the greasing.
\end{itemize}
\item
{\em Bit-slicing} consists in representing an $n$-dimensional vector of $k$-bit sized
coefficients by using $k$ binary vectors of dimension
$n$~\cite{BooBra09}. In particular, one can apply Boolean word instruction
in order to perform arithmetic on 64 dimensional vectors.
\begin{itemize}
\item Over $\F_3$, the binary representation $0 \equiv [0,0], 1\equiv
[1,0], -1 \equiv [1, 1]$ allows 
us to add and subtract two elements in 6 Boolean operations:
$$\begin{array}{ll}
\text{Add}([x_0,x_1],[y_0,y_1]) : & s \leftarrow x_0\oplus y_1 , t \leftarrow x_1 \oplus y_0\\
      & \text{Return} (s\wedge t, (s\oplus x_1) \vee (t\oplus y_1))\\
\text{Sub}([x_0,x_1],[y_0,y_1]) : & t\leftarrow x_0\oplus y_0 \\
                                 & \text{Return}(t\vee (x_1\oplus y_1), (t\oplus
                                 y_1)\wedge(y_0\oplus x_1))
\end{array}
$$
\item Over $\F_5$ (resp. $\F_7$), a redundant representation $x=x_0+2x_1+4x_2 \equiv
[x_0,x_1,x_2]$ allows 
us to add two elements by using  20 (resp. 17) Boolean operations,
negate in 6 (resp. 3) Boolean operations, and double by using  5 (resp. 0) Boolean operations.
\end{itemize}

\begin{table}[htbp]\center
\caption{Boolean operation counts for basic arithmetic by using bit-slicing}
\begin{tabular}{|l|ccc|}
\hline
         & $\F_3$ & $\F_5$ & $F_7$\\
\hline
Addition & 6 & 20 & 17\\
Negation & 1 & 6 &  3\\
Double   &   & 5 &  0\\
\hline
\end{tabular}
\end{table}
\item {\em Bit-packing} consists in representing a vector of field elements
  as an integer 
that 
 fits in a single machine word by using a $2^k$-adic
  representation: $$(x_0,\dots,x_{n-1})\in \F_q^n \equiv
  X=x_0+2^kx_1+\dots +(2^k)^{n-1}x_{n-1} \in \ZZ_{2^{64}}$$  
Elements of extension fields are viewed as
polynomials and stored as the evaluation of this polynomial at the
characteristic of the field. The latter evaluation is called {\em
  Kronecker substitution}~\cite{MR2500374}.
Once we can pack and simultaneously reduce coefficients of
the finite
field in a single machine word, the obtained parallelism can be used
for 
MM. Depending on the respective sizes of the
matrices in the multiplication,
 one can pack only the left operand,
only the right one, or both \cite{DumFouSal11}. 
Then, over 
the field extensions, fast floating point operations can also be used on
the Kronecker substitution of the
elements. 
\end{enumerate}
\end{definition}

All these methods improve in fact the base case of dense linear
algebra, when fast methods are not competitive anymore. 
As already mentioned, the generic cascading idea applies: 
perform 
at first some recursive steps
fast, 
decreasing the matrix
dimension, 
and, at a small enough size, switch to the (improved) classical triple
loop method.

\subsection{Parallelization}
Now, we focus on the design of a parallel 
MM 
routine, 
which computes the matrix product $AB$
based on the \strassenwinograd sequential
algorithm. In order to parallelize the computation at the coarsest
grain, the best approach is 
to apply first  a classical 
block algorithm, generating a prescribed number of independent tasks, 
and  then each of them
will use the sequential \strassenwinograd
algorithm~\cite{Dumas:2014:parPLUQ,Dumas:2015:pfgemm}. 
There, the parallel algorithm is recursive and splits the largest of
either the row dimension of $A$ or the column dimension of $B$, to form
two independent tasks. The granularity of the split is recursive and
terminates whenever the number of created tasks 
becomes larger than the
number of computing resources (e.g., the total number of cores). 
This maximizes the size of the blocks, and therefore the benefit of
\strassenwinograd algorithm, while ensuring a large enough number of
tasks for the computing resources. 

\begin{figure}[ht!]
\centering
\caption{Speed of exact and numerical matrix multiplication routines
  on a 32 cores Intel Xeon E5-4620 2.2Ghz (Sandy Bridge) with 16384KB
  L3 cache~\cite{Dumas:2015:pfgemm}.} 
\label{fig:pfgemmtime}%
\includegraphics[width=\textwidth]{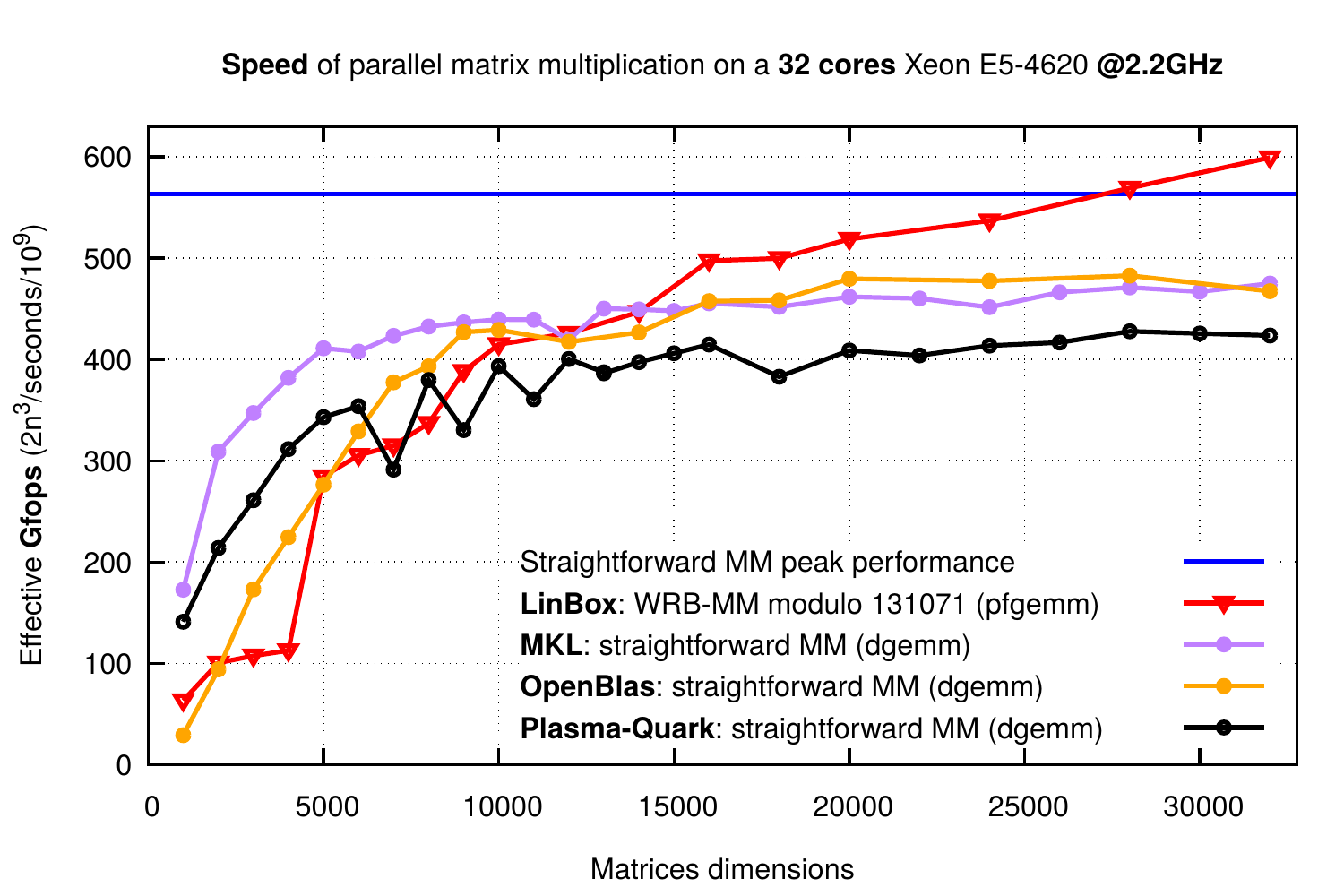}
\end{figure}

Figure~\ref{fig:pfgemmtime} shows the
computation time of various 
MM algorithms: the numerical
\dgemm implementation of \plasmaquark, \openblas and
\texttt{Intel-MKL} as well as  the 
implementation of \pfgemm of \linbox using the \openmp-4.0 data-flow
model. 
Contrary to \texttt{MKL}, \openblas or \plasmaquark, this \pfgemm routine uses the 
above sketched splitting strategy with \strassenwinograd. 
This implementation is run over the finite field $\ZZ/131071\ZZ$
or with real double floating point numbers.  We first notice
that most routines perform very similarly. More precisely,
\texttt{Intel-MKL} \dgemm is faster on small matrices, but the effect
of \strassenwinograd algorithm makes \pfgemm faster on larger
matrices, even 
in the finite field where additional modular reductions
occur.

\section{Some Research Challenges}\label{sprimp}

The field is still in progress, and here are
some current  observations.

(i) The implementations of 
fast trilinear aggregation algorithms by Kaporin in \cite{K99} and
\cite{K04} are relatively little applied in practice so far, 
although they have important advantages versus other algorithms now in use:
being at least as fast and more stable numerically
and allowing highly efficient parallel implementation,
they require much less memory. This is because the algorithms 
of  \cite{K99} and \cite{K04} are defined by 
trilinear decompositions with
{\em supersparse coefficient  matrices} 
$A=(\alpha_{ij}^{(q)})_{(i,j),q}$, $B=(\beta_{jk}^{(q)})_{(j,k),q}$, and $C=(\gamma_{ik}^{(q)})_{(i,k),q}$
 of (\ref{eqcoeff}), which 
is a general property of competently implemented 
trilinear aggregation algorithms. The implementations in 
\cite{K99} and \cite{K04} were intended to be 
initial rather than final steps and were supposed to leave some room for further amelioration.
In particular the algorithm of \cite{P82}, also relying on trilinear aggregation,
has similar features, but supports even a smaller exponent,  and can be a basis 
 for further 
progress in the implementation of fast feasible MM.

(ii) The more recent algorithms of \cite{S13}, obtained by means of computer search, 
are highly promising because they beat the exponent $\log_2(7)$
already for various small problems $MM(m,n,p)$
where $\min\{m,n,p\}=2$ and $\max\{m,n,p\}\le 6$.
 Their 
coefficient matrices $A$, $B$, and $C$ of (\ref{eqcoeff})
are quite dense, however, and their tested performance is inferior to 
recursive bilinear algorithms based on Examples~\ref{ex1}  and~\ref{ex0}.
This leads to the challenge of devising MM algorithms that
would combine the best features of the algorithms of \cite{K99},
\cite{K04}, and \cite{S13}.

(iii) Numerical instability 
of APA algorithms does not prevent them from being efficient in
symbolic computations, but so far only the rudimentary algorithm 
of \cite{BCLR79} has been implemented~\cite{jgd:2015:bini}, while much simpler 
and much more efficient ones are ignored  (cf., e.g., our algorithm 
in Section~\ref{sapa}).

Some researchers in matrix computations still view decreasing the current record 
MM exponent of about 2.37 towards the lower bound 2
as the major theoretical challenge. For the  state of affairs in this area
 we refer the reader to our Figure~\ref{fig2}
and the mostly negative  results in \cite{ASU13}, \cite{AFLG14}, and
\cite{BCCG16}.

We, however,  consider breaking the barrier of 2.7733 for 
the realistic exponent of $MM(n)$, $n\le 1,000,000$,
a more important challenge. The exponent 2.773 stays unbeaten since 1982 (that is, longer than Copper\-smith--Winograd's barrier of 1986,
broken by Stothers in  2010),\footnote{The title of \cite{VW14} is a little deceptive. 
}
and its decrease should require more powerful tools 
than the acceleration of infeasible MM
because of the limitation on the use of recursive bilinear algorithms.
We  hope that
this important challenge will be met in reasonable time, possibly based on
combination of human ingenuity and computer search,\footnote{Computer search has already
helped 
the authors  of \cite{BCLR79}, \cite{Bodrato:2010:SMM}, and \cite{VW14} 
in devising their algorithms.}
and then significant 
impact
on  present day computations should be expected, 
whereas reaching
the exponent 2 for MM of infeasible sizes per se would hardly make such an impact. 

In view of microscopic progress in the decrease of the exponent of 
infeasible MM, the present day research directions towards that goal seem 
to lead to various dead ends,
and any potential progress  shall most likely rely 
on radically new insights, ideas and techniques,
such as aggregation
and mapping the input and output of MM
to another dimension (cf. \cite{P72}, \cite{P82},  \cite{P84} and \cite{LPS92}).
 
In the area of the implementation of MM,
further progress with  recursive
bilinear algorithms based on fast  $2\times 2$ MM  
is a highly important challenge, and the recent advances 
by Bodrato \cite{Bodrato:2010:SMM} and Cenk and Hasan \cite{CH16} 
are very encouraging, but it would greatly benefit  the field if
the researchers  will not confine themselves to the geocentric 
 viewpoint of 1969, restricted to $2\times 2$-based bilinear recursion, 
and will also explore  heliocentric 
point of view of XXI century, by opening themselves to the benefits of
 trilinear world, aggregation, APA, and possibly some other
new powerful Strassen-free techniques for fast feasible MM, yet to appear.

\section*{Acknowledgments}
Jean-Guillaume Dumas's work is partially supported by \OpenDreamKit.
Victor Pan's work has been supported by 
 NSF Grants CCF--1116736 and
 CCF--1563942 and by PSC CUNY Award  68862--00~46.
Furthermore he is grateful to A. Bostan, I.V. Oseledets and E.E. Tyrtyshnikov
for their pointers to recent works on MM and the bibliography 
 on  quantized tensor decompositions, respectively, to
Franklin Lee, Tayfun Pay, and Liang Zhao for their assistance with creating
Figures~\ref{fig1} and~\ref{fig2}, to Igor Kaporin
for sharing his expertise  on 
various issues of practical MM and providing
 information about his papers
\cite{K99} and \cite{K04}, 
and to Ivo Hedtke for his extensive 
 comments to the preprint \cite{P14a}.


\begin{thebibliography}{hspace{0.5in}}

%
%

\bibitem{AHU74}
A.V. Aho, J.E. Hopcroft, J.D. Ullman, 
{\em The Design and Analysis of Algorithms.}
Addison-Wesley, Reading, MA, 1974.

%
%
\bibitem{AlbrechtBH10}
Martin~R. Albrecht, Gregory~V. Bard, and William Hart.
\newblock Algorithm 898: Efficient multiplication of dense matrices over
  {GF(2)}.
\newblock {\em {ACM} Trans. Math. Softw.}, 37(1), 2010.
\newblock \href {http://dx.doi.org/10.1145/1644001.1644010}
  {\path{doi:10.1145/1644001.1644010}}.

%

\bibitem{AGM97}
N. Alon, Z. Galil, O. Margalit, On the Exponent of the All Pairs Shortest Path Problem, 
{J. of Computer and System Sciences}, {\bf 54,~2}, 255--262. 1997.

%

\bibitem{ASU13}
N. Alon, A. Shpilka, C. Umans, On Sunflowers and Matrix Multiplication.
{\em Computational Complexity}, {\bf 22,~2}, 219--243, 2013. 

%

\bibitem{AFLG14}
A. Ambainis, Y. Filmus, F. Le Gall,
Fast Matrix Multiplication: Limitations of the Laser Method.
{\em Electronic Colloquium on Computational Complexity (ECCC)}, year
2014, paper 154.
\url{http://eccc.hpi-web.de/report/2014/154}.

Available at \href{http://arxiv.org/pdf/1411.5414}{\path{arXiv:}1411.5414}, November 21, 2014.
 

%

\bibitem{AP09}
R.R. Amossen, R. Pagh, Faster Join-projects and Sparse Matrix Multiplications. In
{\em Proceedings of the 12th International Conference on Database Theory}, 
121--126, 2009.

%
%
\bibitem{MR0269546}
V.~L. Arlazarov, E.~A. Dinic, M.~A. Kronrod, and I.~A. Farad{\v{z}}ev.
\newblock The economical construction of the transitive closure of an oriented
  graph.
\newblock {\em Doklady Akademii Nauk SSSR}, 194:487--488, 1970.

%

\bibitem{BBDLS15}
Ballard, G., Benson, A. R., Druinsky, A., Lipshitz, B., Schwartz, O.,
 Improving the numerical stability of fast matrix multiplication algorithms, SIMAX, in print, and
\href{http://arxiv.org/pdf/1507.00687}{\path{arXiv:}1507.00687}, 2015.

%

\bibitem{BCD14}
G. Ballard, E. Carson, J. Demmel,
M. Hoemmen, N. Knight, O. Schwartz,
Communication Lower Bounds and Optimal Algorithms 
for Numerical Linear Algebra.
{\em Acta Numerica}, {\bf 23}, 1--155, 2014.

%

\bibitem{BDHS12}
G. Ballard,  J. Demmel,
O. Holtz, O. Schwartz,
 Graph expansion and communication costs of fast matrix multiplication,
{\em Journal of the ACM (JACM)},
{\bf 59, ~6},  
Article No. 32 
doi>10.1145/2395116.2395121

%

\bibitem{BDKSa}
G. Ballard, A. Druinsky, N. Knight, O. Schwartz,
Hypergraph Partitioning for Sparse Matrix-Matrix Multiplication,
\href{http://arxiv.org/pdf/1603.05627}{\path{arXiv:}1603.05627}.

%

\bibitem{B88}
D.H. Bayley,
Extra High Speed Matrix Multiplication on the Cray-2.
{\em SIAM J. on Scientific and Statistical Computing}, {\bf 9,~3}, 603--607, 1988. 

%

%
%
%
%
%

%

\bibitem{BB15}
 A.R. Benson, G. Ballard,
 A framework for practical parallel fast matrix multiplication. 
In {\em Proceedings of the 20th ACM SIGPLAN Symposium 
on Principles and Practice of Parallel Programming}, 42--53, 
ACM Press, New York, January 2015.

%

\bibitem{B80}
D. Bini, 
Relations Between Exact and Approximate Bilinear Algorithms: Applications.
{\em Calcolo},
{\bf 17,~1}, 87--97, 1980.

%

\bibitem{BCLR79}
D. Bini, M. Capovani, G. Lotti, F.  Romani, 
$O(n^{2.7799})$ Complexity for $n\times n$ Approximate Matrix Multiplication.  
{\em Information  Processing Letters}, {\bf 8,~5}, 234--235, June 1979.

%

\bibitem{BL80}
D. Bini, G. Lotti,  
Stability of Fast Algorithms 
for Matrix Multiplication. 
{\em Numerische Math.}, {\bf 36,~1}, 
63--72, 1980.

%

\bibitem{BP86}
D. Bini, V.Y. Pan,
Polynomial Division and Its Computational Complexity.  
{\em J. Complexity}, {\bf 2,~3}, 179--203, 1986. 

%

\bibitem{BP94}
D. Bini, V.Y. Pan,
{\em Polynomial and Matrix Computations, Volume 1: Fundamental Algorithms}.
Birkh\"auser, Boston, 1994.

%

\bibitem{B00}
M. Bl{\"a}ser,
Lower Bounds for the Multiplicative Complexity of Matrix Multiplication.
{\em J. of Computational Complexity}, {\bf 9,~2}, 73--112, 2000.

%

\bibitem{B99}
M. Bl{\"a}ser,
A $5/2 n^{2}$-Lower Bound for the Multiplicative Complexity 
of $n\times n$ Matrix Multiplication over Arbitrary Fields.
{\em Proc. 40th Ann. IEEE Symp. on Foundations of Computer Science (FOCS 1999)}, 
45--50, IEEE Computer Society Press, Los Alamitos, CA
1999.

%

\bibitem{BCCG16}
J. Blasiak, T. Church, H. Cohn, J. A. Grochow, E. Naslund, W. F. Sawin, C. Umans,
On cap sets and the group-theoretic approach to matrix multiplication,
\href{http://arxiv.org/pdf/1605.06702}{\path{arXiv:}1605.06702}, 2016.

%

\bibitem{Bodrato:2010:SMM}
Marco Bodrato.
\newblock A {Strassen}-like matrix multiplication suited for squaring and higher
  power computation.
\newblock In {\em Proceedings of the 2010 International Symposium on Symbolic
  and Algebraic Computation}, ISSAC '10, pages 273--280, New York, NY, USA,
  2010. ACM.
\newblock \href{http://dx.doi.org/10.1145/1837934.1837987}{\path{doi:10.1145/1837934.1837987}}.

%

%
\bibitem{BooBra09}
Thomas~J. Boothby and Robert~W. Bradshaw.
\newblock Bitslicing and the method of four russians over larger finite fields,
  January 2009.
\newblock \href{http://arxiv.org/pdf/0901.1413}{\path{arXiv:}0901.1413}.

%

\bibitem{BM75}
A. Borodin, I. Munro, 
{\em The Computational Complexity of Algebraic and Numeric
Problems.} American Elsevier, New York, 1975.

%

\bibitem{BJS08}
A. Bostan, C.-P. Jeannerod, E. Schost,
Solving structured linear systems 
with large displacement rank.
{\em Theoretical Computer Science}, 
{\bf 407,~1–3}, 155–-181, 2008. 
 	Proceedings version in 
{\em ISSAC'07}, pp. 33–40, ACM Press, NY, 2007.

%

\bibitem{jgd:2015:bini}
B. Boyer and J.-G. Dumas.
\newblock Matrix multiplication over word-size modular fields using approximate
  formulae.
\newblock {\em ACM Transactions on Mathematical Software}, 42(3):20.1--20.12, 2016.
\newblock \url{https://hal.archives-ouvertes.fr/hal-00987812}.

%

\bibitem{14BDGPS}
B. Boyer, J.-G. Dumas, P. Giorgi, C. Pernet, and
  B.~David Saunders.
\newblock Elements of design for containers and solutions in the linbox
  library.
\newblock In Hoon Hong and Chee Yap, editors, {\em Mathematical Software - ICMS
  2014}, volume 8592 of {\em Lecture Notes in Computer Science}, pages
  654--662. Springer Berlin Heidelberg, 2014.
\newblock \href {http://dx.doi.org/10.1007/978-3-662-44199-2_98}
  {\path{doi:10.1007/978-3-662-44199-2_98}}.

%

\bibitem{BDPZ09}
B. Boyer, J.-G. Dumas,
C. Pernet, W. Zhou,
 Memory Efficient Scheduling of Strassen-
Winograd's Matrix Multiplication Algorithm. 
{\em Proc. Intern. Symposium on Symbolic and Algebraic
Computation (ISSAC 2009)},
 55--62, ACM Press, New York, 2009.

%

\bibitem{BD73}
R.W. Brockett, D. Dobkin,
On Optimal Evaluation of a Set of Bilinear Forms.
{\em Proc. of the 5th Annual Symposium on the Theory of Computing (STOC 1973)},
88--95,
ACM Press, New York, 1973.

%

\bibitem{BD76}
R.W. Brockett, D. Dobkin,
On the Number of Multiplications
Required for a Matrix Multiplication.
{\em SIAM Journal on Computing}, {\bf 5,~4}, 624--628, 1976.
%

\bibitem{BD78}
R.W. Brockett, D. Dobkin,
On Optimal Evaluation of a Set of Bilinear Forms.
{\em Linear Algebra and Its Applications}, 
{\bf 19,~3}, 207--235, 1978.

%

\bibitem{B89}
N.H. Bshouty,
 A Lower Bound for Matrix Multiplication.
{\em  SIAM J. on Computing},  {\bf 18,~4}, 759--765, 1989.

%

\bibitem{B95}
N.H. Bshouty,
 On the Additive Complexity of $2\times 2$ Matrix Multiplication, 
{\em Information Processing Letters}, {\bf 56,~6}, 329--335, 1995.

%

\bibitem{BH74}
J.R. Bunch, J.E. Hopcroft,
Triangular Factorization and Inversion by Fast Matrix Multiplication.
{\em Mathematics of Computation},
{\bf 28,~125}, 231--236, 1974.

%

\bibitem{BCS97}
P. B{\"u}rgisser, M. Clausen, M.A. Shokrollahi,
{\em Algebraic Complexity Theory}.
Springer Verlag, 1997.
%

%

\bibitem{CH16}
M. Cenk, M.A. Hasan. 
On the Arithmetic Complexity of Strassen-Like Matrix Multiplications.
{\em J. of Symbolic Computation}, 2016, in press. doi: 10.1016/j.jsc.2016.07.004.

%

\bibitem{CGG84}%
B.W. Char, K.O. Geddes, G.H. Gonnet,
GCDHEU: Heuristic Polynomial GCD Algorithm Based on Integer GCD Computation.
{\em Proceedings of EUROSAM'84, Lecture Notes in Computer Science}, {\bf 174}, 
285--296, Springer, New York, 1984.

%

\bibitem{CKSU05} 
H. Cohn, R. Kleinberg, B. Szegedy, C. Umans,
Group-theoretic Algorithms for
 Matrix Multiplication. 
{\em Proceedings of the 46th Annual
Symposium on Foundations of Computer Science (FOCS 2005)}, (Pittsburgh, PA),
 379--388, IEEE Computer Society Press,
 2005.

%

\bibitem{CU03} 
H. Cohn, C. Umans,
 A Group-theoretic Approach to
Fast Matrix Multiplication. 
{\em Proceedings of the 44th Annual
Symposium on Foundations of Computer Science (FOCS 2003)}, (Cambridge, MA),
 438--449, IEEE Computer Society Press,
 2003.

%

\bibitem{CU13} 
H. Cohn, C. Umans,
 Fast Matrix Multiplication Using Coherent Configurations. 
{\em Proceedings of the 24th Annual ACM-SIAM Symposium on Discrete Algorithms (SODA 2013)}, 
 1074--1087, 2013.

%

\bibitem{C82} 
D. Coppersmith,
 Rapid Multiplication of Rectangular Matrices.
{\em SIAM Journal on Computing}, {\bf 11,~3}, 467--471, 1982.

%

\bibitem{C97} 
D. Coppersmith,
Rectangular Matrix Multiplication Revisited.
{\em Journal of Complexity}, {\bf 13,~1}, 42--49, 1997.

%

\bibitem{CW82} 
D. Coppersmith, S. Winograd,
On the Asymptotic Complexity of Matrix Multiplication.
{\em SIAM J. on Computing}, {\bf 11,~3}, 472--492, 1982. 
Proc. version in {\em 23rd FOCS} 
(Nashville, TN), 82--90, IEEE Computer Society Press, 1981.

%

\bibitem{CW90} 
D. Coppersmith, S. Winograd,
Matrix Multiplicaton via Arithmetic Progressions.
{\em J. of Symbolic Computations}, {\bf 9,~3}, 251--280, 1990.
Proc. version in {\em 19th ACM Symposium on Theory of Computing (STOC 1987)},
 (New York, NY), 1--6,
ACM Press, New York, NY, 1987. Also Research Report RC 12104, 
{\em IBM T.J. Watson Research Center}, August 1986.

%

\bibitem{DN09}  
P. D'Alberto,  A. Nicolau, 
 Adaptive Winograd's Matrix Multiplication.
{\em ACM Transactions on Mathematical Software}, {\bf 36,~1}, paper 3, 2009.

%

\bibitem{DS13}  
 A.M. Davie, A.J. Stothers, Improved Bound for Complexity of Matrix Multiplication.
{\em Proceedings of the Royal Society of Edinburgh}, {\bf 143A}, 351--370, 2013.

%

\bibitem{dG78} 
H.F. de Groot, 
On Varieties of Optimal Algorithms for the 
Computation of Bilinear Mappings. {\it Theoretical Computer Science}, 
{\bf 7,~2}, 127--148, 1978.

%

\bibitem{DI00}
C. Demetrescu, G.F. Italiano, Fully Dynamic Transitive Closure: Breaking Through 
the $O(n^2)$ Barrier. In
{\em Proceedings of the 41st Annual Symposium on Foundations of Computer Science (FOCS 2000)}, 
381--389, 2000.

%

\bibitem{DDH07} 
 J. Demmel, I. Dumitriu,
    O. Holtz, 
Fast Linear Algebra Is Stable.
{\em Numerische Mathematik},
 {\bf 108,~1}, 59--91, 2007. 

%

\bibitem{DDHK07}
 J. Demmel, I. Dumitriu, O. Holtz, R. Kleinberg,
Fast Matrix Multiplication Is Stable.
{\em Numerische Mathematik},
 {\bf 106,~2}, 199--224, 2007.

%
%
\bibitem{Dixon:1982:Pad}
John~D. Dixon.
\newblock {Exact solution of linear equations using p-adic expansions.}
\newblock {\em Numerische Mathematik}, 40:137--141, 1982.

%

\bibitem{DHSS94}
 C.C. Douglas, M. Heroux, G. Slishman, 
R.M. Smith,
GEMMW: A Portable Level 3 BLAS Winograd Variant Of 
Strassen's Matrix-Matrix Multiply Algorithm.
{\em J. of Computational Physics},
{\bf 110,~1}, 1--10, 1994.
 
%

\bibitem{DIS11} 
C.-E. Drevet, Md. N. Islam, {\'E}. Schost,
Optimization Techniques for Small Matrix Multiplication. 
{\em Theoretical Computer Science}, {\bf 412,~22}, 2219--2236, 2011.


%

\bibitem{DP09}  
R. Duan, S. Pettie, Fast Algorithms for  $(\max-\min)$ Matrix Multiplication
and Bottleneck Shortest Paths, 
{\em Proceedings of the 15th Annual ACM-SIAM Symposium on
Discrete Algorithms (SODA 2009)}, 384--391, 2009.

%
%

\bibitem{MR2500374}
J.-G. Dumas.
\newblock Q-adic transform revisited.
\newblock In {\em Proceedings of the 2008 {I}nternational {S}ymposium on
  {S}ymbolic and {A}lgebraic {C}omputation}, pages 63--69, New York, 2008. ACM.
\newblock \href {http://dx.doi.org/10.1145/1390768.1390780}
  {\path{doi:10.1145/1390768.1390780}}.

\bibitem{DumFouSal11}
J.-G. Dumas, Laurent Fousse, and Bruno Salvy.
\newblock Simultaneous modular reduction and {Kronecker} substitution for small
  finite fields.
\newblock {\em Journal of Symbolic Computation}, 46(7):823 -- 840, 2011.
\newblock Special Issue in Honour of Keith Geddes on his 60th Birthday.
\newblock \href {http://dx.doi.org/10.1016/j.jsc.2010.08.015}
  {\path{doi:10.1016/j.jsc.2010.08.015}}.

\bibitem{Dumas:2002:icms}
J.-G. Dumas, T. Gautier, M. Giesbrecht, P. Giorgi, B. Hovinen,
E. Kaltofen, B. D. Saunders, W. J. Turner and G. Villard.
\newblock {LinBox}: A Generic Library for Exact Linear Algebra.
\newblock In {\em {ICMS}'2002, Proceedings of the 2002 International Congress of
  Mathematical Software}, pages 40--50, Beijing, China.
\newblock
\url{http://ljk.imag.fr/membres/Jean-Guillaume.Dumas/Publications/icms.pdf}.

\bibitem{jgd:2002:fflas}
J.-G. Dumas, T. Gautier, and C. Pernet.
\newblock Finite field linear algebra subroutines.
\newblock In Teo Mora, editor, {\em {ISSAC}'2002, Proceedings of the 2002 ACM
  International Symposium on Symbolic and Algebraic Computation, Lille,
  France}, pages 63--74. ACM Press, New York, July 2002.
\newblock
  \url{http://ljk.imag.fr/membres/J.-G..Dumas/Publications/Field_blas.pdf}.

\bibitem{Dumas:2015:pfgemm}
J.-G. Dumas, T. Gautier, C. Pernet, J.-L. Roch, and
  Z. Sultan.
\newblock Recursion based parallelization of exact dense linear algebra
  routines for {Gaussian} elimination.
\newblock {\em Parallel Computing}, 57:235--249, 2016.
\newblock \url{http://hal.archives-ouvertes.fr/hal-01084238}.

\bibitem{Dumas:2014:parPLUQ}
J.-G. Dumas, T. Gautier, C. Pernet, and Z. Sultan.
\newblock Parallel computation of echelon forms.
\newblock In {\em {Euro-Par 2014}, Proceedings of the 20th international
  conference on parallel processing, Porto, Portugal}, volume 8632 of {\em
  Lecture Notes in Computer Science}, pages 499--510, August 2014.
\newblock \url{http://hal.archives-ouvertes.fr/hal-00947013}.

\bibitem{DGP04}
 J.-G. Dumas, P. Giorgi, C. Pernet,
FFPACK: Finite Field Linear Algebra
Package.
{\em Proc. Intern. Symposium on Symbolic and Algebraic
Computation (ISSAC 2004)},
Jaime Gutierrez, editor, 119--126, ACM Press, New York, 2004.

\bibitem{DGP08}
J.-G. Dumas, P. Giorgi, and C. Pernet.
\newblock Dense linear algebra over prime fields.
\newblock {\em ACM Transactions on Mathematical Software}, 35(3):1--42,
  November 2008.
\newblock \url{http://hal.archives-ouvertes.fr/hal-00018223}.

\bibitem{Dumas:2012:HFF}
J.-G. Dumas and C. Pernet.
\newblock Computational linear algebra over finite fields.
\newblock In Daniel Panario and Gary~L. Mullen, editors, {\em Handbook of
  Finite Fields}. Chapman \& Hall/CRC, 2012.
\newblock \url{http://hal.archives-ouvertes.fr/hal-00688254}.

\bibitem{DPS16}
J.-G. Dumas, C. Pernet, Z. Sultan,
Fast Computation of the Rank Profile Matrix and the Generalized Bruhat
Decomposition, {\em J. of Symbolic Computation}, 2016.
Proc. version in {ISSAC 2015}.

%
%

\bibitem{F72}
C.M. Fiduccia,
On Obtaining Upper Bound on the Complexity of
Matrix Multiplication.
In {\em Analytical Complexity of Computations}
 (edited by R.E. Miller, J. W. Thatcher, J. D. Bonlinger),
in  the {\em the IBM Research Symposia Series}, 
pp. 31--40, Plenum Press, NY, 1972. 

%

\bibitem{F72a}
C.M. Fiduccia,
Polynomial Evaluation via the Division Algorithm: 
The Fast Fourier Transform Revisited.
{\em Proc. 4th Annual ACM Symposium 
on Theory of Computing (STOC 1972)},
88--93, ACM Press, New York, 1972. 

%

\bibitem{FP74}
 M.J. Fischer,  M.S. Paterson,  
String-Matching and
Other Products.
\textit{SIAM--AMS Proc.}, {\bf 7}, 113--125, 1974.

%

\bibitem{F74}
P.C. Fischer,
Further Schemes for Combining Matrix Algorithms.
{\em Proceedings of the 2nd Colloquium on Automata, Languages and Programming},
{\em Lecture Notes in Computer Science}, {\bf 14}, 
428--436, 
Springer-Verlag, London, UK, 1974.

%

\bibitem{F09}
M. F{\"u}rer,
Faster Integer Multiplication.
{\em SIAM J. on Computing}, {\bf 39,~3}, 979--1005, 2009. 

%

\bibitem{GG13}
J. von zur Gathen, J. Gerhard (2013). 
{\em Modern Computer Algebra}. Cambridge University Press, Cambridge, UK, 
third edition,
2013.

%

\bibitem{GL13}
G.H. Golub, C.F. Van Loan,
{\em Matrix Computations}.
Johns Hopkins University Press, Baltimore, Maryland, 2013 (4th addition).

%

\bibitem{GKT13}
L. Grasedyck, D. Kressner, C. Tobler, 
A Literature Survey of Low-rank Tensor Approximation Techniques.
{\em GAMM-Mitteilungen},
{\bf 36}, {\bf 1}, 53--78, 2013.

%

\bibitem{H90}
N.J. Higham,
Exploiting Fast Matrix Multiplication within Level 3 BLAS.
{\em ACM Trans. on Math. Software}, {\bf 16,~4}, 352--368, 1990.

%
%
\bibitem{HoevenLQ14}
Joris~{van der} Hoeven, Gr{\'{e}}goire Lecerf, and Guillaume Quintin.
\newblock Modular {SIMD} arithmetic in mathemagix.
\newblock ACM Transactions on Mathematical Software, 43(1):5, August 2016.
\newblock \href{http://arxiv.org/abs/1407.3383}{\path{arXiv:}1407.3383}.

\bibitem{HongKung81}
J.~W. Hong and H.~T. Kung.
\newblock {I/O} complexity: The red-blue pebble game.
\newblock In {\em Proc. 14th {STOC}}, pages 326--333, New York, NY, USA, 1981.
  ACM.
\newblock \href {http://dx.doi.org/10.1145_800076.802486}
  {\path{doi:10.1145_800076.802486}}.

%

\bibitem{HK69}
J.E. Hopcroft, L.R. Kerr,
 Some Techniques for Proving Certain Simple Programs Optimal. 
{\em Proceedings of the Tenth Annual Symposium on Switching and Automata Theory},
 36--45, IEEE Computer Society Press, 1969.

%

\bibitem{HK71}
J.E. Hopcroft, L.R. Kerr, 
On Minimizing the Number of Multiplications Necessary for Matrix Multiplication.
{\em SIAM J. on Applied Math.}, {\bf 20,~1}, 30--36, 1971.

%

\bibitem{HM73}
J.E. Hopcroft, J. Musinski,
Duality Applied to Matrix Multiplication and Other
Bilinear Forms. {\em SIAM Journal on Computing}, {\bf 2,~3},
159--173, 1973.

%

\bibitem{HP98}
X. Huang, V.Y. Pan,
Fast Rectangular Matrix Multiplication and Applications.
{\em Journal of Complexity}, {\bf 14,~ 2}, 257--299, 1998.
Proc. version in  
{\em Proc. Annual ACM International Symposium on Parallel Algebraic and Symbolic Computation (PASCO'97)}, 11--23, ACM Press, New York, 1997.

%
\bibitem{HSHG16}
J. Huang, T.M. Smith, G.M. Henry, R.A. van de Geijn,
Strassen's Algorithm Reloaded.
{\em Proceedings of the International Conference for High Performance
Computing, Networking, Storage and Analysis}, {SC '16},
%
59:1--59:12, IEEE Press, Piscataway, NJ, USA, 2016.
\url{http://www.computer.org/csdl/proceedings/sc/2016/8815/00/8815a690.pdf}
%
\bibitem{HussLederman:1996:ISA}
Steven Huss-Lederman, Elaine~M. Jacobson, Jeremy~R. Johnson, Anna Tsao, and
  Thomas Turnbull.
\newblock Implementation of {Strassen}'s algorithm for matrix multiplication.
\newblock In {ACM}, editor, {\em Supercomputing '96 Conference Proceedings:
  November 17--22, Pittsburgh, {PA}}, New York, NY 10036, USA and 1109 Spring
  Street, Suite 300, Silver Spring, MD 20910, USA, 1996. ACM Press and IEEE
  Computer Society Press.
\newblock \href {http://dx.doi.org/10.1145/369028.369096}
  {\path{doi:10.1145/369028.369096}}.

\bibitem{Ibarra:1982:LSP}
Oscar~H. Ibarra, Shlomo Moran, and Roger Hui.
\newblock A generalization of the fast {LUP} matrix decomposition algorithm and
  applications.
\newblock {\em Journal of Algorithms}, 3(1):45--56, March 1982.
\newblock \href {http://dx.doi.org/10.1016/0196-6774(82)90007-4}
  {\path{doi:10.1016/0196-6774(82)90007-4}}.

\bibitem{Jeannerod201346}
C.-P. Jeannerod, C. Pernet, and A. Storjohann.
\newblock Rank-profile revealing {Gaussian} elimination and the {CUP} matrix
  decomposition.
\newblock {\em Journal of Symbolic Computation}, 56(0):46 -- 68, 2013.
\newblock %
%
\href{http://dx.doi.org/10.1016/j.jsc.2013.04.004}{\path{doi:10.1016/j.jsc.2013.04.004}}.

\bibitem{JML86}
R. W. Johnson, A. M. McLoughlin, 
Noncommutative Bilinear Algorithms
for $3\times 3$ Matrix Multiplication, 
{\em SIAM J. on Computing}, {\bf 15,~2}, 
595--603, 1986.

%

\bibitem{Kaltofen:2011:EACM}
E. Kaltofen and A. Storjohann.
\newblock {\em Encyclopedia of Applied and Computational Mathematics}, Chapter
  ``Complexity of computational problems in exact linear algebra",
227--233.
\newblock Springer, November 2015.
\newblock
\href{http://dx.doi.org/10.1007/978-3-540-70529-1_173}{\path{doi:10.1007/978-3-540-70529-1_173}}.
%

%
%
\bibitem{KSV06}
H. Kaplan, M. Sharir, E. Verbin, 
Colored Intersection Searching via Sparse
 Rectangular Matrix Multiplication. In 
{\em Proceedings of the 22nd ACM Symposium on Computational Geometry}, 52--60, 2006.

%

\bibitem{K99}
I. Kaporin,
A Practical Algorithm for Faster Matrix Multiplication.
{\em Numerical Linear Algebra with Applications}, {\bf 6,~8}, 687--700, 1999.

%

\bibitem{K04}
I. Kaporin,
The Aggregation and Cancellation Techniques as a 
Practical Tool for Faster Matrix Multiplication. 
{\it Theoretical Computer Science}, 
{\bf 315}, {\bf 2--3}, 469--510, 2004.

%

\bibitem{KZHP08}
 S. Ke, B. Zeng, W. Han,
V. Y. Pan, Fast Rectangular Matrix Multiplication and Some
Applications.
{\em Science in China, Series A: Mathematics}, {\bf 51,~3}, 389--406, 2008.

%
\bibitem{Keller-Gehrig:1985:FAC}
Walter Keller-Gehrig.
\newblock Fast algorithms for the characteristic polynomial.
\newblock {\em Theor. Comput. Sci.}, 36(2-3):309--317, June 1985.
\newblock \url{http://dl.acm.org/citation.cfm?id=3929.3939}.

%

\bibitem{K11}
B.N. Khoromskij,
$O(d\log N)$ Quantics Approximation of $N$-$d$ Tensors in High-dimensional 
Numerical Modeling.
{\em Constructive Approximation},
{\bf 34}, {\bf 2}, 257--280,
2011.

%

\bibitem{K97}
D.E. Knuth,
{\em The Art of Computer Programming: Volume 2, Seminumerical Algorithms}.
Addison-Wesley, Reading, Massachusetts, 1969 (first edition),
 1981 (second edition), 1997 (third edition).

%

\bibitem{KB09}
T.G. Kolda, B.W. Bader, 
Tensor Decompositions and Applications.
{\em SIAM Review}, {\bf 51,~3}, 455--500, 2009.

%

\bibitem{LPS92}
J. Laderman, V.Y. Pan, H.X. Sha, 
On Practical Algorithms for Accelerated Matrix Multiplication.
{\em Linear Algebra and Its Applications}, 
{\bf 162--164}, 557--588, 1992.

%

%
%
%
%
%
%
%

%

\bibitem{L14}
J.M. Landsberg,  
New Lower Bound for the Rank of Matrix Multiplication.
{\em SIAM J. on Computing}, {\bf 43,~1}, 144--149, 2014.

%

\bibitem{LG12}
F. Le Gall,
 Faster Algorithms for Rectangular Matrix Multiplication. 
{\em Proceedings of the 53rd
Annual IEEE Symposium on Foundations of Computer Science (FOCS 2012)},
 514--523,  IEEE Computer Society Press, 2012.

%

\bibitem{LG14}
F. Le Gall,
Powers of Tensors and Fast Matrix Multiplication.
{\em Proceedings of the 39th International Symposium on Symbolic and Algebraic Computation (ISSAC 2014)}, 
296--303, ACM Press, New York, 2014.

%

\bibitem{L02}
L. Lee,
 Fast Context-free Grammar Parsing Requires Fast Boolean Matrix Multiplication. 
{\em Journal of the ACM (JACM)}, {\bf  49,~1}, 1--15, 2002.

%

\bibitem{LR83}
 G. Lotti, F.  Romani,
On the Asymptotic Complexity of Rectangular Matrix Multiplication.
{\em Theoretical Computer Science}, {\bf 23}, 171--185, 1983.

%

%
%
%
%

%

\bibitem{MR14}
A. Massarenti, E. Raviolo,
Corrigendum to "The rank of
 $n\times n$
matrix multiplication is at least $3n^2-2\sqrt 2n^{3/2}-3n$"  
[{\em Linear
Algebra and its Applications}, {\bf 438,~11} (2013) 4500--4509].
{\em Linear
Algebra and its Applications}, 
{\bf 445}, 369--371, 2014.

%

\bibitem{MP80}
W. L. Miranker, V. Y. Pan,
Methods of Aggregations, 
{\em Linear Algebra and Its Applications}, 
{\bf 29}, 231--257, 1980.          

%

\bibitem{M55}
T.S. Motzkin,
Evaluation of Polynomials and Evaluation of Rational Functions.
{\em Bull. of Amer. Math. Society}, {\bf 61,~2}, 163, 1955.

%

\bibitem{N16}
V. Neiger,
Fast computation of shifted Popov forms of polynomial 
matrices via systems of modular polynomial equations.
{\em Proceedings of the International Symposium on Symbolic and Algebraic 
Computation (ISSAC'16)}, 365--372, ACM Press, New York, 2016.

%

\bibitem{NS16}
A. Neumaier, D. Stehl\'e,
Faster LLL-type reduction of lattice bases, 
{\em Proceedings of the International Symposium on Symbolic and Algebraic 
Computation (ISSAC'16)}, 373--380, ACM Press, New York, 2016.

%

\bibitem{NS16a}
J.S.R. Nielsen,  A. Storjohann,
Algorithms for Simultaneous Pad\'e Approximation,
{\em Proceedings of the International Symposium on Symbolic and Algebraic 
Computation (ISSAC'16)}, 405--412, ACM Press, New York, 2016.

%

\bibitem{O09} 
I.V. Oseledets, 
Approximation of Matrices with Logarithmic Number of Parameters. 
{\em Dokl. Math.}, {\bf 80,~2}, 653--654, 2009.

%

\bibitem{O10} 
I.V. Oseledets, 
Approximation of $2^d \times 2^d$ Matrices Using Tensor Decomposition.
{\em SIAM J. on Matrix Analysis and Applications},
{\bf 31}, {\bf 4}, 2130--2145, 2010.

%

\bibitem{OT10} 
I.V. Oseledets, E.E. Tyrtyshnikov,
TT-cross Approximation for Multidimensional Arrays.
{\em Linear Algebra Appls.} {\bf 432,~1}, 70--88, 2010.


%

\bibitem{OT11} 
I.V. Oseledets, E.E. Tyrtyshnikov,
Algebraic Wavelet Transform via Quantics Tensor Train Decomposition.
{\em SIAM J. Scientific Computing}, {\bf 33,~3}, 1315--1328, 2011.

%

\bibitem{O54} 
A.M. Ostrowski,
On Two Problems in Absract Algebra Connected with Horner's Rule.
In the {\em Studies Presented to R. von Mises}, 40--48,
Academic Press, New York, 1954. 

%

\bibitem{P66}
V.Y. Pan, 
On Methods of Computing the Values of Polynomials.
{\em Uspekhi Matematicheskikh Nauk},
{\bf 21}, {\bf 1(127)}, 103--134, 1966.
[Transl. {\em Russian Mathematical Surveys},
{\bf 21, 1(127)}, 105--137, 1966.]

%

\bibitem{P72}
V.Y. Pan,
On Schemes for the Evaluation of Products and Inverses of Matrices (in Russian). 
{\em Uspekhi Matematicheskikh Nauk}, {\bf 27}, {\bf 5} {\bf (167)}, 249--250, 1972.

%

\bibitem{P78}
V.Y. Pan,
Strassen's Algorithm Is Not Optimal. Trilinear Technique of Aggregating for  
Fast Matrix Multiplication. {\em Proc. the 19th Annual IEEE Symposium on 
Foundations of Computer Science (FOCS'78)}, 
166--176, IEEE Computer Society Press, Long Beach, California, 1978.

%

\bibitem{P79}
V.Y. Pan,
Fields Extension and Trilinear Aggregating, 
Uniting and Canceling for the Acceleration of Matrix Multiplication. 
{\em Proceedings of the 20th Annual IEEE Symposium on Foundations 
of Computer Science (FOCS'79)},
28--38, IEEE Computer Society Press, Long Beach, California, 1979. 

%

\bibitem{P80}
V.Y. Pan,
New Fast Algorithms for Matrix Operations. 
{\em SlAM J. on Computing}, {\bf 9,~2}, 321--342, 1980,
and Research Report RC 7555, {\em IBM T.J. Watson Research Center}, February 1979. 

%

\bibitem{P80a}
V.Y. Pan,
The Bit-Operation Complexity of the Convolution of Vectors and of the DFT.
Technical report 80-6, Computer Science Dept., SUNY, Albany, NY, 1980.
(Abstract in Bulletin of EATCS, {\bf 14}, page 95, 1981.)
 
%

\bibitem{P81}
V.Y. Pan,
New Combinations of Methods for the Acceleration of Matrix Multiplications.
{\em Computers and Mathematics (with Applications)}, {\bf 7,~1}, 73--125, 1981.

%

\bibitem{P82}
V.Y. Pan,  Trilinear Aggregating with Implicit Canceling for a New Acceleration of Matrix Multiplication.
{\em Computers and Mathematics (with Applications)}, {\bf 8,~1}, 23--34, 1982.

%

\bibitem{P84}
V.Y. Pan, How Can We Speed up Matrix Multiplication?
{\em SIAM Review,} {\bf 26}, {\bf 3}, 393--415, 1984.

%

\bibitem{P84a}
V.Y. Pan,
Trilinear Aggregating and the Recent Progress in the Asymptotic 
Acceleration of Matrix Operations. 
{\em Theoretical Computer Science}, {\bf 33,~1}, 117--138, 1984. 

%

\bibitem{P84b}
V.Y. Pan, 
{\em How to Multiply Matrices Faster.}
{\em Lecture Notes in Computer Science}, 
{\bf 179}, Springer, Berlin, 1984.

%

%
%
%
%

%

\bibitem{P14}
V.Y. Pan, 
Better Late Than Never: Filling a Void in the History
  of Fast Matrix Multiplication and Tensor Decompositions,
\href{http://arxiv.org/pdf/1411.1972}{\path{arXiv:}1411.1972}, November 2014.

%

\bibitem{P14a}
V.Y. Pan, 
Matrix Multiplication, Trilinear Decompositions, APA Algorithms, and Summation,
\href{http://arxiv.org/pdf/1412.1145}{\path{arXiv:}1412.1145} CS, Submitted December 3, 2014, revised February 5, 2015.

%

\bibitem{pernet:tel-01094212}
C. Pernet.
\newblock {\em {High Performance and Reliable Algebraic Computing}}.
\newblock Habilitation {\`a} diriger des recherches, {Universit{\'e} Joseph
  Fourier, Grenoble 1}, November 2014.
\newblock \url{https://tel.archives-ouvertes.fr/tel-01094212}.

%

\bibitem{P16}
C. Pernet,
Computing with Quasiseparable Matrices,
{\em Proceedings of the International Symposium on Symbolic and Algebraic 
Computation (ISSAC'16)}, 389--396, ACM Press, New York, 2016.

%

\bibitem{Pernet:2007:FAC}
C. Pernet and A. Storjohann.
\newblock Faster algorithms for the characteristic polynomial.
\newblock In {\em Proceedings of the 2007 International Symposium on Symbolic
  and Algebraic Computation}, ISSAC '07, pages 307--314, New York, NY, USA,
  2007. ACM.
\newblock \href{http://dx.doi.org/10.1145/1277548.1277590}{\path{doi:10.1145/1277548.1277590}}.

%

\bibitem{P76}
R. L. Probert,
On the Additive Complexity of Matrix Multiplication.
{\em SIAM J. on Computing}, {\bf 5,~2}, 187--203, 1976.

%

\bibitem{P74}
R.L. Probert,
On the Complexity of Symmetric Computations.
{\em Canadian J. of Information Processing and Operational Res.}, {\bf 12,~1}, 71--86, 1974.

%

\bibitem{RS03}
R. Raz, A. Shpilka,
  Lower Bounds for Matrix Product, in Bounded Depth Circuits with Arbitrary Gates.
  {\em SIAM J. on Computing}, {\bf 32,~2},
  488--513, 2003.

%

\bibitem{R79}
F. Romani, 
private communication at the Oberwolfach Conference in October of 1979.

%

\bibitem{R82}
F. Romani, 
Some Properties of Disjoint Sum of Tensors Related to MM.
{\em SIAM J. on Computing}, {\bf 11,~2}, 263--267, 1982.

%

\bibitem{SM10}
P. Sankowski, M. Mucha,
Fast Dynamic Transitive Closure with Lookahead.
{\em Algorithmica}, {\bf 56,~ 2}, 180--197, 2010.

%

\bibitem{S81}
A. Sch{\"o}nhage,
 Partial and Total Matrix Multiplication.
{\em SIAM J. on Computing}, {\bf 10,~3}, 434--455, 1981. 

%

\bibitem{S82}
A. Sch{\"o}nhage,
Asymptotically Fast Algorithms for the
Numerical Multiplication and
  Division of Polynomials with Complex Coefficients.
\textit{Proc. EUROCAM, Marseille}
  (edited by  J. Calmet),
  \textit{Lecture Notes in Computer Science} \textbf{144}, 
3--15, Springer, Berlin, 1982.

%

\bibitem{SS71}
A. Sch{\"o}nhage, V. Strassen, 
Schnelle Multiplikation gro\ss{}er Zahlen. 
{\em Computing}, {\bf 7,~3--4}, 281--292, 1971.

%

\bibitem{S13}
A.V. Smirnov, 
The Bilinear Complexity and Practical Algorithms for Matrix Multiplication. 
{\em Computational Mathematics and Mathematical Physics}, 
{\bf 53,~12}, 1781--1795 (Pleiades Publishing, Ltd), 2013. 
Original Russian Text 
in {\em Zhurnal Vychislitel'noi Matematiki i Matematicheskoi Fiziki}, 
{\bf 53,~12}, 1970--1984, 2013.

%

\bibitem{S03}
A. Shpilka,
Lower Bounds for Matrix Product.
  {\em SIAM J. on Computing},
{\bf 32,~5}, 1185--1200, 2003.

%

\bibitem{S15}
A. Storjohann,
On the complexity of inverting integer and polynomial matrices,
{\em Computational Complexity}, {\bf 24}, 777--821, 2015.  

%

%
\bibitem{Storjohann:2005:HighOrder}
A. Storjohann.
\newblock The shifted number system for fast linear algebra on integer
  matrices.
\newblock {\em Journal of Complexity}, 21(4):609--650, 2005.
\newblock \href {http://dx.doi.org/10.1016/j.jco.2005.04.002}
  {\path{doi:10.1016/j.jco.2005.04.002}}.

\bibitem{SY15}
 A. Storjohann, S. Yang, 
 A relaxed algorithm for online matrix inversion,
{\em Proceedings of the International Symposium on Symbolic 
and Algebraic Computation (ISSAC'15)},
 339-346, ACM Press, New York, 2015.

%

\bibitem{S10}
A.J. Stothers,
On the Complexity of Matrix
Multiplication. Ph.D. Thesis, University of Edinburgh, 2010.

%

\bibitem{S69}
 V. Strassen, Gaussian Elimination Is Not Optimal. 
{\em Numerische Math.}, {\bf 13}, 354--356, 1969.

%

\bibitem{S72}
V. Strassen,
Evaluation of Rational Functions, in {\em Analytical Complexity of
Computations} (edited by R.E. Miller, J. W. Thatcher, and J. D. Bonlinger), 
pages 1--10, Plenum Press, New York, 1972. 

%

\bibitem{S73}
 V. Strassen, 
Vermeidung von Divisionen.
{\em J. Reine Angew. Math.}, {\bf 1973,~ 264}, 184--202, 1973.

%

\bibitem{S74}
V. Strassen,
Some Results in Algebraic Complexity Theory, in
{\em Proceedings of the International Congress of
Mathematicians}, Vancouver, 1974 (Ralph D. James, editor), Volume {\bf 1}, pages 497--501,
 Canadian Mathematical Society
1974. 


%

\bibitem{S86}
 V. Strassen, The Asymptotic Spectrum of Tensors and
the Exponent of Matrix Multiplication. 
{\em Proc. 27th Ann. Symposium on Foundation of Computer Science}, 
49--54, 1986.

%

\bibitem{T55}
J. Todd,
Motivation for Working in Numerical Analysis.
{\em Communication on Pure and Applied Math.}, {\bf 8,~1}, 97--116, 1955.

%

\bibitem{T03}
E.E. Tyrtyshnikov, Tensor Approximations of Matrices 
Generated by Asymptotically Smooth Functions. 
{\em Mat. Sbornik}, {\bf 194,~6}, 147--160, 2003.

%

\bibitem{VW14}
V. Vassilevska Williams,
Multiplying Matrices Faster than Coppersmith--Winograd.
Version available at 
\url{http://theory.stanford.edu/virgi/matrixmult-f.pdf}, retrieved on
January 30, 2014. Also see
{\em Proc. 44th Annual ACM Symposium on Theory of Computing
(STOC 2012)}, 887--898, ACM Press, New York, 2012.

%

%
%
%
%

%

\bibitem{W70a}
 A. Waksman, 
On Winograd's Algorithm for Inner Products.
{\em IEEE Transactions on Computers}, {\bf C-19,~4}, 360--361, 1970.
 
%

\bibitem{W67}
S. Winograd, 
On the Number of Multiplications Required to Compute Certain Functions.
{\em Proc. of the National Academy of Sciences}, {\bf 58,~5}, 1840--1842, 1967.

%

\bibitem{W68}
S. Winograd, 
A New Algorithm for Inner Product.
{\em IEEE Transaction on Computers},
{\bf C--17,~7}, 693--694, 1968.

%

\bibitem{W70}
   S. Winograd,
On the Number of Multiplications Necessary to Compute Certain Functions.
{\em Communications on Pure and Applied Mathematics},
{\bf 23,~2}, 165--179, 1970.

%

\bibitem{W71}
   S. Winograd,
On  Multiplication of $2\times 2$ Matrices.
{\em Linear Algebra and Its Applications}, 
{\bf 4}, 382--388, 1971.

%

\bibitem{W80}
S. Winograd, 
{\em Arithmetic Complexity of Computations}.
CBMS-NSF Regional Conference Series in Applied Math., {\bf 33},
SIAM, Philadelphia, 1980.

%

\bibitem{Y09}
R. Yuster,
Efficient Algorithms on Sets of Permutations, Dominance, and Real-weighted APSP 
{\em Proceedings of the 15th Annual ACM-SIAM Symposium on
Discrete Algorithms (SODA 2009)}, 384--391, 2009.

%

\bibitem{YZ04}
R. Yuster, U. Zwick, U. Detecting Short Directed Cycles Using Rectangular 
Matrix Multiplication and Dynamic Programming. In
{\em Proceedings of the 15th Annual ACM-SIAM Symposium on
Discrete Algorithms (SODA 2004)}, 254--260, 2004.

%

\bibitem{YZ05}
R. Yuster, U. Zwick, Fast Sparse Matrix Multiplication.
{\em ACM Transactions on Algorithms}, {\bf 1,~1}, 2--13, 2005.

%

\bibitem{YZ05a}
R. Yuster, U. Zwick,
 Answering Distance Queries in Directed Graphs Using Fast Matrix Multiplication. In
{\em Proceedings of 
46th Annual IEEE Symposium on the Foundations of Computer Science (FOCS 2005)},
389--396, IEEE Computer Society Press, 2005.

%

\bibitem{Z02}
U. Zwick. All-pairs Shortest Paths Using Bridging Sets and Rectangular Matrix Multiplication.
{\em Journal of the ACM}, {\bf 49,~3}, 289--317, 2002.

%

\end{thebibliography}
\end{document}